\documentclass[11pt]{article}
\usepackage{graphicx} 
\usepackage{amsmath, amssymb}
\usepackage[a4paper, left=1in, right=1in, top=1in, bottom=1in]{geometry}
\usepackage{appendix}
\usepackage{xcolor}
\usepackage{braket}
\usepackage{times}

\usepackage{amsmath, amssymb, amsthm}
\usepackage[
  colorlinks=true,
  breaklinks=true,
  pdfusetitle=true,
  hypertexnames=false,
  allcolors=blue]{hyperref}
\usepackage{graphicx}
\usepackage{comment}
\usepackage{mathtools}
\usepackage[giveninits=true,maxbibnames=99,style=alphabetic,maxalphanames=4,minalphanames=3,backend=biber]{biblatex}
\usepackage{mathrsfs}
\usepackage{algorithm}
\usepackage[noend]{algpseudocode}
\usepackage{float}
\usepackage[section]{placeins}
\usepackage{thm-restate}

\usepackage{accents}

\AtEveryBibitem{\clearlist{language}}
\renewbibmacro*{doi+eprint+url}{%
  \iftoggle{bbx:doi}
    {\printfield{doi}}
    {}%
  \newunit\newblock
  \ifboolexpr{togl {bbx:eprint} and test {\iffieldundef{doi}}}
    {\usebibmacro{eprint}}
    {}%
  \newunit\newblock
  \ifboolexpr{togl {bbx:url} and test {\iffieldundef{doi}}  and test {\iffieldundef{eprint}}}
    {\usebibmacro{url+urldate}}
    {}}

\newcommand*{\Tr}{\operatorname{Tr}}

\DeclarePairedDelimiterX\ketbra[2]{\lvert}{\rvert}{#1 \delimsize\rangle\delimsize\langle #2}

\newcommand{\SQ}[1]{\textup{SQ}(#1)}

\providecommand{\myvec}[1]{\ensuremath{\boldsymbol{#1}}}

\providecommand{\ll}{\ensuremath{\myvec{l}}}

\providecommand{\calU}{\ensuremath{\mathcal{U}}}

\newcommand{\poly}{\textup{poly}}
\renewcommand{\exp}{\textup{exp}}

\mathchardef\mhyphen="2D 

\DeclarePairedDelimiter\norm{\lVert}{\rVert}
\newcommand{\fnorm}[1]{\norm{#1}_{\textup F}}

\usepackage[capitalise]{cleveref}
\newtheorem{theorem}{Theorem}[section]
\newtheorem{definition}[theorem]{Definition}
\newtheorem{proposition}[theorem]{Proposition}

\newtheorem{corollary}[theorem]{Corollary}

\newtheorem{lemma}[theorem]{Lemma}
\newtheorem{remark}[theorem]{Remark}
\newtheorem*{remark*}{Remark}

\usepackage{enumitem}
\setlist[itemize]{itemsep=1pt, topsep=2pt}
\setlist[enumerate]{itemsep=1pt, topsep=2pt}

\newcommand{\cfont}{\mathsf}
\newcommand{\MCSV}{\cfont{MCSV}}
\newcommand{\CSV}{\cfont{CSV}}
\newcommand{\CSVhkp}{\cfont{CSV}_{\cfont{HKP}}}
\newcommand{\CSVgc}{\cfont{CSV}_{\cfont{GC}}}
\newcommand{\CSVmyz}{\cfont{CSV}_{\cfont{MYZ}}}
\newcommand{\SampleCSV}{\cfont{SampleCSV}}
\newcommand{\QMA}{\cfont{QMA}}
\newcommand{\QCMA}{\cfont{QCMA}}
\newcommand{\LH}{\cfont{LH}}
\newcommand{\CLDM}{\cfont{CLDM}}
\newcommand{\iDLH}{1D\mhyphen\LH}
\newcommand{\iDCLDM}{1D\mhyphen\CLDM}
\newcommand{\SuperQMA}{\cfont{SuperQMA}}
\newcommand{\SuperQMAPoly}{\cfont{SuperQMA}_{\textup{poly}}}
\newcommand{\SuperQMAExp}{\cfont{SuperQMA}_{\textup{exp}}}

\newcommand{\qc}{\cfont{qc\text{-}\Sigma_2}}
\newcommand{\lmin}{\lambda_{\textup{min}}}
\newcommand{\density}{\mathscr{D}}
\newcommand{\Exp}{\mathbb E}
\newcommand{\ObsCon}{\cfont{ObsCon}}
\newcommand{\ObsConPoly}{\cfont{ObsCon}_{\textup{poly}}}
\newcommand{\ObsConExp}{\cfont{ObsCon}_{\textup{exp}}}
\newcommand{\QMAt}{\cfont{QMA(2)}}
\newcommand{\Pqc}{\cfont{\qc(2)}}

\newcommand{\ProductObsCon}{\cfont{ProdObsCon}}
\newcommand{\PSPACE}{\cfont{PSPACE}}
\newcommand{\QS}{\cfont{Q\Sigma_2}}
\newcommand{\QMAplus}{\cfont{QMA^+}}
\newcommand{\PSQMA}{\cfont{\SuperQMA(2)}}
\newcommand{\MTObsCon}{\cfont{BLOC}}
\newcommand{\SymQMA}{\cfont{SymQMA(k)}}
\newcommand{\Cl}{\mathrm{Cl}}
\newcommand{\E}{\mathbb{E}}
\newcommand{\F}{\mathbb{F}}
\newcommand{\C}{\mathbb{C}}
\newcommand{\SimSnap}{\ensuremath{\mathsf{Sim}^{\mathrm{snap}}_{\widetilde V^{(s)}}}}
\newcommand{\negl}{\cfont{negl}}
\newcommand{\SimInt}{\ensuremath{\mathsf{Sim}^{\mathrm{Int}}_{\widetilde V^{(s)}}}}
\newcommand{\SimHNN}{\ensuremath{\mathsf{Sim}^{\mathrm{HNN}}_{\widetilde V^{(s)}}}}
\DeclarePairedDelimiter\abs{\lvert}{\rvert}%

\title{How hard is it to verify a classical shadow?}
\author{Georgios Karaiskos\footnotemark[1]\and Dorian Rudolph\footnotemark[2]\and Johannes Jakob Meyer\footnotemark[3]\and Jens Eisert\footnotemark[4]\and Sevag Gharibian\footnotemark[2]}

\date{}

\begin{document}
\begingroup
\renewcommand\thefootnote{\fnsymbol{footnote}}
\maketitle
\footnotetext[1]{Department of Computer Science, Paderborn University, Germany. Email: georgios.karaiskos@upb.de.}
\footnotetext[2]{Department of Computer Science and Institute for Photonic Quantum Systems (PhoQS), Paderborn University, Germany. Email: \{dorian.rudolph, sevag.gharibian\}@upb.de.}
\footnotetext[3]{Freie Universität Berlin, Germany. Email: jj.meyer@outlook.com.}
\footnotetext[4]{Freie Universität Berlin, Germany. Email: jense@zedat.fu-berlin.de.}
\endgroup

\begin{abstract}
    Classical shadows are succinct classical representations of quantum states which allow one to encode a set of properties $P$ of a quantum state $\rho$, while only requiring measurements on logarithmically many copies of $\rho$ in the size of $P$. In this work, we initiate the study of verification of classical shadows, denoted classical shadow validity (CSV), from the perspective of computational complexity,  which asks: Given a classical shadow $S$, how hard is it to verify that $S$ predicts the measurement statistics of a quantum state? We first show that even for the elegantly simple classical shadow protocol of [Huang, Kueng, Preskill, Nature Physics 2020] utilizing local Clifford measurements, CSV is QMA-complete. This hardness continues to hold for the high-dimensional extension of said protocol due to [Mao, Yi, and Zhu, PRL 2025]. In contrast, we show that for the HKP and MYZ protocols utilizing global Clifford measurements, CSV can be ``dequantized'' for low-Frobenius norm observables, i.e., solved in randomized poly-time with standard sampling assumptions. Finally, we show that CSV for exponentially many observables is complete for a quantum generalization of the second level of the polynomial hierarchy, yielding the first natural complete problem for such a class. 
\end{abstract}

\section{Introduction}\label{scn:intro}

Fully classically describing a quantum state $\rho$ has long been known to require exponential overhead, making characterizing the outputs of quantum devices a challenging task. Indeed, for a state $\rho$ on $n$ qubits, i.e., of dimension $D=2^n$, a sample complexity of $\Theta(D^2)$ copies of $\rho$ are known to be necessary and sufficient for full quantum state tomography~\cite{OW15,haahSampleoptimalTomographyQuantum2016}. In general, however, one is not necessarily interested in learning \emph{everything} about $\rho$, but only a specific set $P$ of properties. Formally, we may model these as a set of $M$ measurement operators $P=\set{P_i}_{i=1}^M$, where one is interested in computing $\Tr(\rho P_i)$. The natural question is now: \emph{Can one avoid full state tomography in this case?}

In 2018, Aaronson showed~\cite{aaronsonShadowTomographyQuantum2018} the answer is \emph{yes}: for set $P$ of $2$-outcome measurements, given $k$ copies of $\rho$, one can produce estimates $b_1,\ldots, b_M\in[0,1]$ such that with probability at least $1-\delta$, one has $\abs{\Tr(P_i\rho)-b_i}\leq \epsilon$ \emph{for all} $i$. The magic here is that the \emph{sample complexity}, $k$, can be chosen polylogarithmic in the dimension $D$ and number of measurements $M$, i.e., 
\begin{equation}
    k\in O\left( \log \frac{1}{\delta} \cdot \epsilon^{-5} \cdot \log^4 M \cdot \log D \right).
\end{equation}
While this original protocol was not yet \emph{time} efficient, it did not take long for the latter to be rectified, e.g., Brand$\tilde{\textup{a}}$o, Kalev, Li, Lin, Svore, Wu~\cite{brandaoQuantumSDPSolvers2019}.
Indeed, soon after \emph{Huang, Kueng and Preskill} (HKP) discovered~\cite{HKP20} a remarkably simple and efficient classical shadow tomography procedure, which for a given $P$, randomly samples unitary $U$ from an ``appropriate'' ensemble $\cal U$ of unitaries, and measures $U\rho U^\dagger$ in the standard basis. Roughly, the resulting string can be thought of as a ``snapshot'' of $\rho$, and the set of all snapshots constitutes the classical shadow, $S$. A recovery procedure via median-of-means is then specified, so that given $S$, one can recover estimates for $\Tr(\rho M_i)$. For general $P$ and $\calU$, the procedure has sample complexity
\begin{equation}\label{eqn:shadowbound}
    O \left( \frac{\log(M)}{\epsilon^2} \max_{1 \leq i \leq M} \left\| P_i - \frac{\Tr(P_i)}{2^n} I \right\|_{\text{shadow}}^2 \right),
\end{equation}  
where the shadow norm depends on $P$ and $\calU$ (see \Cref{scn:preliminaries} for details on the HKP protocol.) For $\calU$ the set of global Cliffords and the set of $k$-local Cliffords, the shadow norm in \Cref{eqn:shadowbound} is at most $3\Tr(P_i^2)$ and $4^k \|P_i\|_{\infty}^2$, respectively. Thus, for example, to predict measurement results for $P$ the set of $k$-local Pauli strings\footnote{A Pauli string $Q$ is an element of $\set{I,X,Y,Z}^n$. We say $Q$ is $k$-local if it contains at most $k$ non-identity terms.}, one obtains a sample \emph{and} time efficient\footnote{Since $k$-local Clifford measurements $\calU$ are easy to implement.} 
protocol, which requires only $\log(M)$ copies, and as a bonus needs only measure a single copy at a time. 

\paragraph{Verifying classical shadows.} This work initiates the study of the natural question: 
\begin{quote}
    \emph{Given as input a ``classical shadow'' $S$, what is the complexity of verifying that $S$ actually ``predicts'' the measurement statistics of some $\rho$ against $P$?}
\end{quote}
As stated, this question is ill-posed, in the sense that we are not aware of a formal definition of a ``classical shadow'' in the literature. Thus, to remedy this, we first provide a general formal definition:

\begin{restatable}{definition}{defshadow}(Classical shadow)\label{def:classical-shadow}
    A shadow on $n$ qubits is a $4$-tuple $(S,O,A, \chi)$, where 
    \begin{itemize}
        \item (Shadow) $S=\set{s_i}_{i=1}^N$ is a multi-set of $\poly(n)$-bit strings, with $N=\poly(n)$,
        \item (Observables) $O=\set{O_i}_{i=1}^m$ is a set of $n$-qubit observables satisfying $\|O_i\|_{\infty}\le1$, where $1\leq m\leq 2^{p(n)}$ for polynomial $p$. Given index $i$, a $\poly(n)$-bit description of $O_i$ can be produced in $\poly(n)$-time\footnote{This is the succinct access assumption.}. Moreover, there exists a $\poly(n)$-time quantum algorithm which, for any $O_i$ and any $n$-qubit state $\rho$, applies\footnote{Formally, we can efficiently measure in the eigenbasis of $O_i$, and return the eigenvalue corresponding to the measurement result.} measurement $O_i$ to $\rho$.
        \item (Recovery algorithm) $A$ is a $\poly(n)$-time classical algorithm which, given $S$ and $i\in[m]$, produces real number $A(S,i)\in[-1,1]$ within $\chi$ bits of precision.
    \end{itemize}    
\end{restatable}
\noindent This definition says nothing about prediction accuracy; it simply formalizes the idea that a ``classical shadow'' is a multi-set of strings $S$, \emph{in principle} obtained via some set of efficient measurements on copies of a physical state $\rho$, coupled with a set of target observables $O$ and an efficient recovery procedure $A$ for ``extracting predictions''. An alternate possible definition might be not to give shadow $S$ as a fixed sequence of strings, but rather to generate $S$ on-the-fly by sampling from some unknown distribution (thus capturing the idea of measurement bases $\calU$ as in HKP). To capture this, we also define a ``sampled classical shadow'' in \Cref{scn:variant}, and show the complexity of verifying ``classical shadows'' (\Cref{def:classical-shadow}) versus ``sampled classical shadows'' (\Cref{def:sampled}) is equivalent under randomized reductions. For simplicity, we thus work with \Cref{def:classical-shadow}, as its input model is the standard one used in (e.g.) BQP and QMA.

Moving on, the task of checking the \emph{validity} of a shadow, i.e., that the outputs of $A$ correctly predict measurement statistics, is formalized as:

\begin{restatable}{definition}{defCSV}(Classical Shadow Validity ($\CSV$))\label{def:CSV}
Given classical shadow $(S,O,A,\chi)$, parameters $\alpha$ and $\beta$ satisfying $\beta-\alpha\geq 1/\poly(n)$, decide between the following two cases:
\begin{itemize}
    \item \textbf{Yes}: $\exists$  $n$-qubit state $\rho$ s.t.\ $\forall$ $i\in [m]$, 
    $\left| \Tr\left(O_i\rho\right)-A(S,i)\right|\leq \alpha$.
    \item \textbf{No}: $\forall$ $n$-qubit states $\rho$ $\exists$ some $i\in [m]$ s.t.\ $\left| \Tr \left(O_i\rho\right)-A(S,i)\right|\geq \beta$.
\end{itemize}
Since $A(S,i)\in[-1,1]$ and $\|O_i\|_{\infty}\le1$, it is natural to assume $0\le\alpha<\beta\le2$.
\end{restatable}
\noindent The theme of this work is to characterize the complexity of this problem and its variants, including for the HKP protocol with local Clifford measurements, ``dequantization'' results for global Clifford measurements, and the case of exponentially many observables. 

\paragraph{Comparison to and distinction from CONSISTENCY problem.} Before proceeding, the reader familiar with quantum complexity theory may notice that, at least in the setting of \emph{polynomially} many observables $O_i$, CSV is eerily similar to the QMA-complete CONSISTENCY problem of Liu~\cite{liuConsistencyLocalDensity2006}. In the latter, the input is a set of $k$-local reduced states $\rho_i$ acting on a subset $S_i$ of $k\in O(1)$ out of $n$ qubits each, and the question is whether there exists an $n$-qubit state $\rho$ such that for all $i$, $\Tr_{[n]\setminus S_i}(\rho)\approx \rho_i$? Indeed, as our definition of classical shadows is intentionally very general, it includes as a special case the CONSISTENCY problem. From this, one immediately obtains that CSV is at least QMA-hard (\Cref{cor:CSVpoly}). This is \emph{not} the point of this paper! 

The point is that classical shadow protocols used in practice typically do \emph{not} produce local density operators as in CONSISTENCY, but rather highly non-local snapshots (e.g. $n$-local operators which are the tensor product of $n$ non-trivial single qubit states)! Our goal is thus to characterize the complexity of CSV for precisely these experimentally relevant snapshots, to which the QMA-hardness of CONSISTENCY does not obviously apply. Indeed, as will be discussed shortly, a reduction from CONSISTENCY to CSV will have to overcome the well-known challenging problem of how to construct a \emph{global} quantum snapshot from \emph{local} overlapping reduced density operators (as in CONSISTENCY). 

\paragraph{Motivation and application to near-term devices.} (1) As near-term experimental devices remain noisy, verifying that a device \emph{actually} outputs the state intended remains a major challenge. Well-known examples include verification of quantum advantage experiments such as Random Circuit Sampling~\cite{boixoCharacterizingQuantumSupremacy2018} or Boson Sampling~\cite{aaronsonComputationalComplexityLinear2011,hamiltonGaussianBosonSampling2017}. In the same vein, it is arguably important to verify the validity of snapshots output by classical shadow experiments; this is true even if the experimenter fully trusts that the device has not been tampered with. More generally, in the distributed cloud setting where the user does \emph{not} trust the device, CSV becomes yet more crucial. 

(2) The study of $\CSV$ is important to the study of classes $\QMA$ versus $\cfont{QCMA}$ (i.e. $\QMA$ with a classical proof \cite{AV02}), as it underpins the central open question of whether quantum witnesses are fundamentally more powerful than classical ones. Specifically, \emph{if} a QMA verifier's measurement falls into a class of observables $O$ whose output statistics could be efficiently predicted by poly-size classical shadows, \emph{and} if CSV for said shadows could be solved by a (uniformly generated) poly-time quantum circuit, then $\QMA=\QCMA$, which would be a breakthrough.

(3) We further motivate the study of $\CSV$ by framing it as a natural quantum analogue of the classical Sparse Representation problem ($\cfont{SR}$) (see \Cref{def:SR}) under the lens of classical shadow protocols. In $\cfont{SR}$, one has to decide if a sparse vector, consistent with a given measurement sketch, exists. While $\cfont{SR}$ is $\cfont{NP}$-hard in the general case \cite{FR13}, it becomes efficiently solvable (via convex relaxation) when the measurement matrix satisfies the Restricted Isometry Property (RIP) \cite{CRT06}. In the quantum setting an analogous compressed sensing phenomenon is known: under a low-rank promise on the state and suitable measurements, convex programs can efficiently reconstruct the state \cite{GLFBE}. Our work studies the quantum SR question arising in the classical shadow framework, \cite{aaronsonShadowTomographyQuantum2018, HKP20}. 

\paragraph{Our results.} We organize our discussion\footnote{We remark that although we gave a fully general formal definition of classical shadows (\Cref{def:classical-shadow}), most of our results are actually independent of the specific recovery algorithm $A$ employed therein; thus, behind the scenes we often work with a simpler restatement of CSV, denoted Observable Consistency ($\ObsCon$, \Cref{def:obscon}). Hence, while we informally state our results in terms of CSV here, our formal statements are often in terms of $\ObsCon$.} in terms of (1) polynomially many observables, (2) exponentially many observables, and (3) further variants of $\CSV$ with connections to $\QMAt$. For clarity, our main results involve (1) and (2). All hardness results are under poly-time many-one reductions.\\ 

\vspace{-1mm}
\noindent\emph{1. Polynomially many observables: Hardness and dequantization.} As previously stated, it is not difficult to see that $\CSV$ in its most general form is QMA-complete (\Cref{cor:CSVpoly}). Here, we focus on the more challenging case of the HKP protocol~\cite{HKP20} (instantiated with either local or global Clifford measurements), as well as a high-dimensional generalization thereof due to Mao, Yi and Zhu (MYZ) for odd-prime local dimension $d$~\cite{MYZ25}. 

To begin, we define $\CSVhkp$ as CSV for the HKP protocol instantiated with \emph{local} Clifford measurements (\Cref{def:hkpcsv}); roughly, the elements of $S$ are $n$-bit strings, conjugated by Pauli strings in $\set{X,Y,Z}^n$, and the observables are $k$-local Pauli strings for $k\in O(1)$. We show the following statement. 

\begin{theorem}[Informal; see \Cref{cor:CSVHKPisQMAcomplete}]\label{thm:csvhkpinformal}
    $\CSVhkp$ is $\QMA$-complete, even for $6$-local observables on a spatially sparse hypergraph.
\end{theorem}
\noindent In words, deciding if a given HKP classical shadow based on local Clifford measurements is valid is intractable, even when the observables are $6$-local and essentially arranged on a line (formally on a spatially sparse hypergraph (in the sense of Ref.\ \cite{oliveirarobertoComplexityQuantumSpin2008}; \Cref{def:sparse})). Specifically, the hardness construction may be viewed as 1D nearest-neighbor on qudits of dimension $8$. Each qudit is then decomposed into $3$ qubits, and neighboring pairs of qudits $(q_i,q_{i+1})$ have a $6$-local observable acting jointly on their constituent qubits. 

Defining $\CSVmyz$ (\Cref{def:csvmyz}) analogously for MYZ on odd prime local dimensions $d$, we next show:  
\begin{theorem}[Informal; see \Cref{cor:CSVMYZisQMAcomplete}]
    $\CSVmyz$ is $\QMA$-complete for a fixed odd prime local dimension $d\geq 11$, even for $2$-local nearest-neighbor observables on a line.
\end{theorem}
\noindent Here, since we are allowed to work with larger $d$, we cleanly obtain hardness with all observables acting on pairs of nearest neighbor qudits $(q_i,q_{i+1})$.

Finally, we study CSV for the HKP protocol instantiated with \emph{global} Clifford measurements, denoted $\CSVgc$ (\Cref{def:CSV_GC}). We show a ``dequantization'' result as follows, for Frobenius norm $\fnorm{A}=\sqrt{\Tr(A^\dagger A)}$:
\begin{theorem}[Informal (see \Cref{thm:csvgc})]\label{thm:csvgcinformal}
    $\CSVgc$ is solvable in polynomial classical randomized time if (a) $\|O_i\|_F\le \poly(n)$ for all $i$, and (b) we are given sampling and query access to each $O_i$.
\end{theorem}    
\noindent First, while the Frobenius norm bound above may \emph{a priori} seem strong, this setting captures natural tasks such as (e.g.) fidelity estimation against pure and low rank states~\cite{HKP20}. Moreover, the \emph{global} Clifford HKP protocol itself is efficient precisely in this regime, i.e. when $\|O_i\|_F\le \poly(n)$, coinciding with the regime in which \Cref{thm:csvgcinformal} dequantizes $\CSVgc$.  
By a similar argument, this ``dequantization" result extends to the MYZ protocol with \emph{global} Clifford measurements, under the same assumptions (\cref{thm:csvgcmyz}). Second, we dub this a ``dequantization'' result, in that conditions (a) and (b) are similar to those in previous dequantization works, e.g., Tang~\cite{tangQuantuminspiredClassicalAlgorithm2019} and Chia, Gily\'en, Li, Lin, Tang and Wang~\cite{CGLLTW22}, allowing randomized linear algebra techniques to be employed. Specifically, our definition of sampling and query access is from \cite{CGLLTW22}.\\

\noindent\emph{2. Exponentially many observables.} We next consider CSV with \emph{exponentially} many observables. Although \emph{a priori} this setting may seem unrealistic, King, Gosset, Kothari and Babbush gave~\cite{kingTriplyEfficientShadow2025} an explicit \emph{polynomial}-sample complexity shadow protocol for the set of all $4^n$ Pauli string observables, $\set{I,X,Y,Z}^n$. (Note the time complexity is still exponential, but in our setting, we do not produce the shadow, but receive it as input; thus, this overhead is not relevant.) What is also relevant is that Ref.\  \cite{kingTriplyEfficientShadow2025} gives a \emph{poly}-time recovery algorithm for the observable expectations, \emph{assuming} one only demands \emph{constant} additive error. We show the following statement.

\begin{theorem}[Informal; follows from \Cref{l:obsconconstant} and \Cref{cor:qc=SQMA_exp}]\label{thm:csvexpinformal}
    CSV for exponentially many observables and constant additive error recovery precision is $\qc$-complete.
\end{theorem}
\noindent Let us discuss strengths and weaknesses: The strengths are that (1) the result holds even if one need only recover constant precision approximations 
of observable predictions, and (2) \Cref{thm:csvexpinformal} yields the first natural complete problem for a quantum generalization of (a level of) the \emph{polynomial hierarchy} (PH). Specifically, $\qc$ (\Cref{def:qc}) is a quantum generalization of $\Sigma_2^p$, the second level of PH, in which the first proof is a mixed quantum state, the second a classical string, and the verifier is quantum. We remark this is the first work studying $\qc$, though other variants of quantum PH have been studied in previous works~\cite{gharibianHardnessApproximationQuantum2012,gharibianQuantumGeneralizationsPolynomial2018,agarwalQuantumPolynomialHierarchies2024,grewalEntangledQuantumPolynomial2024}. The weakness is that, unlike \Cref{thm:csvhkpinformal}, we do not prove the result for the specific observable set of Ref.~\cite{kingTriplyEfficientShadow2025}, i.e., 
for $\set{I,X,Y,Z}^n$. \\ 

\vspace{-1mm}
\noindent\emph{3. Further variants of $\CSV$ and connections to $\QMAt$.} For completeness, we also show the following for variants of $\CSV$: 
\begin{enumerate}
    \item (\Cref{scn:product}) CSV where the consistent state must be \emph{product}, i.e., $\rho=\rho_A\otimes\rho_B$, is $\QMAt$-complete\footnote{$\QMAt$ is QMA, but where the proof is promised to be in tensor product~\cite{kobayashiQuantumMerlinArthurProof2003}.} for polynomially many observables, and $\cfont{qc\text{-}\Sigma_2(2)}$-complete (\cref{def:P-qc-sigma2}) for exponentially many observables. As an intermediate step, the proof shows that $\PSQMA_{\poly}=\QMAt$, where $\PSQMA_{\poly}$ is a product state generalization of $\QMAplus$ from Ref.~\cite{aharonovLatticeProblemQuantum2003}. 

    \item Verifying if a \emph{set} of classical shadows, each possibly with different observables, all correspond to the \emph{same} state $\rho$ (\Cref{def:MCSV}) is QMA-complete (\Cref{cor:mcsvpoly_qma}) and $\qc$-complete (\Cref{cor:mcsvexp_qc}) for polynomially many and exponentially many observables, respectively.
\end{enumerate}

\paragraph{Techniques.} We focus on QMA-completeness of $\CSVhkp$ (\Cref{thm:csvhkpinformal}), ``dequantization" of low-Frobenius norm $\CSVgc$ (\Cref{thm:csvgcinformal}), and $\qc$-completeness of CSV with exponentially many observables (\Cref{thm:csvexpinformal}). \\

\vspace{-1mm}
\noindent \emph{QMA-completeness of $\CSVhkp$.} Ideally, we wish to reduce the QMA-complete CONSISTENCY problem on $k$-local reduced density operators $\rho_i$ to $\CSVhkp$. The challenge? Each $\rho_i$ acts on only $k$-qubits. HKP classical shadows with local Clifford measurements, on the other hand, have shadow elements $s_i$ which are highly \emph{non-local} --- each $s_i$ is an $n$-qubit tensor product of eigenvectors of single-qubit Pauli matrices. And ``stitching'' together local information, i.e., the $\rho_i$, to obtain \emph{globally} consistent information, i.e., the $s_i$, is a difficult task, reminiscent of the quantum marginal problem. 

To overcome this requires a series of steps. First, we start with the 1D CONSISTENCY problem on qudits, so that it suffices to stitch together \emph{nearest neighbor} reduced states $(\rho_{i,i+1},\rho_{i+1,i+2})$ on the line. To this end, we first show\footnote{One could alternatively use the 1D CONSISTENCY QMA-hardness result of Liu~\cite{liuLocalConsistencyProblem2007}, but this would only yield hardness under Turing reductions, not many-one reductions.} QMA-completeness of 1D CONSISTENCY with local dimension $d=8$ via many-one reduction by combining the locally simulatable technique of Broadbent and Grilo~\cite{BG22} with the QMA-complete result for the 1D Local Hamiltonian problem with $d=8$ of Hallgren, Nagaj, and Narayanaswami~\cite{HNN13}. Then, we take the local nearest-neighbor reduced states $\rho_i$ on qu-$8$-its from 1D CONSISTENCY, decompose each qudit $q_i$ into a triple of qubits $T_i$, and consider all possible \emph{$6$-local} HKP shadows on pairs $(T_i,T_{i+1})$. Crucially, we know under the HKP protocol that any valid local shadow's expectation should exactly recover the corresponding state $\rho_i$. Using this fact and our 1D setup, one can derive a \emph{linear program}  (LP) which captures ``how much weight/probability'' to put onto each local shadow, so that the ``local probabilities'' are consistent with some global HKP shadow if and only if the 1D CONSISTENCY instance we started with is a YES instance.

Unfortunately, solving this LP is itself not enough, because we next need to simulate the probability of a local shadow $s_j$ occurring when measuring local Cliffords in HKP by \emph{repeating} $s_j$ an appropriate integer number of times in our shadow set $S$. We must, in fact, do this \emph{exactly} to ensure consistency, and so we next ``round'' our LP into an \emph{integer program}  (IP) to give us \emph{integer} weights on local shadows. This raises the potential roadblock that solving integer programs is NP-hard, but here we again crucially use the fact that we are working in 1D. Specifically, we  exploit the 1D structure to design an efficient dynamic program to solve the IP, obtaining the desired integer weights on local shadows. Finally, we construct a list of \emph{global} shadows by repeatedly carefully stitching together strings of local shadows under appropriate permutations given by a perfect matching.\\

\vspace{-1mm}
\noindent\emph{Dequantization of low-Frobenius-norm $\CSV_{\mathrm{GC}}$.}
For global Clifford shadows, we work with the abstract observable
consistency problem, $\ObsCon_{\mathrm{F,Samp}}$, whose input consists of
pairs $(O_i,y_i)$ with $\|O_i\|_F\le \poly(n)$, together with additional sampling and
query access assumptions to the observables. The consistency conditions can then be
viewed as an SDP feasibility problem with bounded-Frobenius-norm constraint
matrices. We invoke the randomized classical SDP solver of~\cite{CGLLTW22},
whose running time is polynomial in the Frobenius bound, the inverse promise
gap, the number of constraints, and $\log N$ for system dimension $N=2^n$.
This yields a randomized polynomial-time algorithm for $\CSV_{\mathrm{GC}}$
under the stated sampling and query assumptions. The same proof applies to
the global-Clifford MYZ protocol after replacing $N=2^n$ by $N=d^n$; for
fixed local dimension $d$, one has $\log N=n\log d=O(n)$, so the runtime
remains polynomial in the input size.\\

\vspace{-1mm}
\noindent\emph{$\qc$-completeness of CSV with exponentially many observables.} To connect $\CSV$ with $\qc$, we use the $\QMAplus$ formalism of
Aharonov and Regev~\cite{aharonovLatticeProblemQuantum2003}, in which one is
given polynomially many measurements $\Pi_i$ and targets $r_i\in \mathbb{R}$,
and asks whether there is a state $\rho$ such that
$\Tr(\Pi_i\rho)\approx r_i$ for all $i$. While
Ref.~\cite{aharonovLatticeProblemQuantum2003} showed $\QMAplus=\QMA$, here
we define the analogous class with \emph{exponentially} many measurements
$\Pi_i$, denoted $\SuperQMAExp$. We prove that $\CSV$ with exponentially many
observables is $\SuperQMAExp$-complete, and subsequently show that
$\SuperQMAExp=\qc$. One direction encodes a $\qc$ verifier as an
exponentially indexed family of consistency checks. Conversely, for
$\SuperQMAExp\subseteq\qc$, the existential quantum proof provides the
candidate globally consistent state, while the universal classical proof
specifies one of the exponentially many checks to test. The remaining
technical ingredient is amplification for $\qc$: in the NO case, Sion's
minimax theorem gives a single distribution over classical challenges that is
bad against every quantum state; a matrix Hoeffding bound sparsifies this
distribution to a polynomial-size list of challenges, after which standard
weak-error amplification applies.

\paragraph{Open questions.} We have initiated the study of the complexity of verifying classical shadows. For hardness, an important open question is whether other specific classical shadow protocols and observable sets have QMA-hard CSV problems? In the case of exponentially many observables, for example, can one give a $\qc$-completeness proof of CSV for the protocol of King, Gosset, Kothari and Babbush~\cite{kingTriplyEfficientShadow2025}? The main bottleneck we faced here was that, unlike in our proof for HKP with polynomially many observables (\Cref{thm:csvhkpinformal}), it is not clear how to start from an ``exponential size'' analogue of the QMA-complete 1D CONSISTENCY problem. A natural idea might be to start with \emph{translationally invariant} 1D systems~\cite{gottesmanQuantumClassicalComplexity2009}. Such systems, however, act on \emph{exponentially} many qudits, whereas our setting requires \emph{polynomially} many qubits --- the exponentiality occurs only in the number of observables for CSV. Finally, are there instances of CSV aside from our HKP, MYZ global Clifford results which can also be dequantized, or even better, solved classically without sampling assumptions? 

\paragraph{Organization.} \Cref{scn:preliminaries} begins with preliminaries, including reviews of the HKP and MYZ classical shadow protocols. \Cref{scn:obscon} studies the general CSV problem (i.e., not restricted to any particular shadow protocol), including the case of exponentially many observables. \Cref{scn:HKPhardness} studies the HKP and MYZ protocols, showing QMA-hardness and our dequantization result. \Cref{scn:product} studies product variants of CSV, and \Cref{scn:variant} the sampled and multiple shadow consistency variant of CSV.

\section{Preliminaries}\label{scn:preliminaries}

\paragraph{Definitions.} We use $A\leq B$ and $A\leq_r B$ to denote poly-time deterministic many-one and poly-time randomized reductions from $A$ to $B$, respectively.

\begin{definition}[Sparse Representation (\textsf{SR})]
\label{def:SR}
The input of the Sparse Representation problem is a tuple $(\Psi, v, m, \epsilon)$, where:
\begin{enumerate}
    \item \textbf{$\Psi$}: An $n \times d$ real-valued matrix, known as the \emph{measurement matrix}, where $n \ll d$. The matrix is given by a succinct description from which its entries can be computed in polynomial time.
    \item \textbf{$v$}: An $n$-dimensional real-valued vector, known as the \emph{measurement sketch}.
    \item \textbf{$m$}: A positive integer, known as the \emph{sparsity parameter}.
    \item \textbf{$\epsilon$}: A non-negative real number representing the allowed \emph{error tolerance}.
\end{enumerate}
Decide between the following 2 cases:
\begin{itemize}
    \item \textbf{Yes:} $\exists\; f\in \mathbb{R}^d$ such that $\|f\|_0\le m\; \text{and}\;   \| \Psi \mathbf{f} - v \|_2 \le \epsilon$.
    \item \textbf{No:} $\forall\; f\in \mathbb{R}^d\;\;\text{with}\; \|f\|_0\le m \;\text{it holds that}\;\| \Psi \mathbf{f} - v \|_2 > \epsilon$. 
\end{itemize}
\end{definition}

\begin{definition}[Spatial sparsity~\cite{oliveirarobertoComplexityQuantumSpin2008}]\label{def:sparse}
  A spatially sparse hypergraph $G$ on $n$ vertices has:
  \begin{enumerate}
  \item every vertex participates in $O(1)$ hyper-edges, and
  \item there is a straight-line drawing in the plane such that every hyper-edge overlaps with $O(1)$ other hyper-edges and the surface covered by every hyper-edge is $O(1)$.
  \end{enumerate}
\end{definition}

\begin{definition}[QMA with unentangled provers (QMA(2))]
A promise problem $\mathcal{A}=(A_{\mathrm{yes}},A_{\mathrm{no}},A_{\mathrm{inv}})$ is in $\mathrm{QMA}(2)$ if there exists a P-uniform quantum circuit family $\{V_n\}$ and polynomials $p,q:\mathbb{N}\to\mathbb{N}$ satisfying the following properties. For any input $x\in\{0,1\}^n$, the verifier $V_n$ takes in $n+2p(n)+q(n)$ qubits as input, consisting of the input $x$ on register $X$, a quantum proof $\ket{\psi_1}_A\otimes\ket{\psi_2}_B \in \big((\mathbb{C}^2)^{\otimes p(n)}\big)^{\otimes 2}$ on registers $A\otimes B$, and $q(n)$ ancilla qubits initialized to $\ket{0}$ on register $C$. The first qubit of register $C$, denoted $C_1$, is the designated output qubit, a measurement of $C_1$ in the standard basis after applying $V_n$ yields the following:
\begin{itemize}
    \item (Completeness) If $x\in A_{\mathrm{yes}},\;\exists$ proof $\ket{\psi_1}_A\otimes\ket{\psi_2}_B \in \big((\mathbb{C}^2)^{\otimes p(n)}\big)^{\otimes 2}$ such that $V_n$ accepts with probability $\ge 2/3$.
    \item (Soundness) If $x\in A_{\mathrm{no}}$, then $\forall$ proofs $\ket{\psi_1}_A\otimes\ket{\psi_2}_B \in \big((\mathbb{C}^2)^{\otimes p(n)}\big)^{\otimes 2}$, $V_n$ accepts with probability $\le 1/3$.
\end{itemize}
\end{definition}

\begin{definition}[$\QMA^+$\cite{aharonovLatticeProblemQuantum2003}]
A language $L \in \mathrm{QMA}^+$ if there exists a super-verifier\footnote{A ``super-verifier" is given by a classical polynomial-time randomized algorithm that given an input $x$ outputs a description of a quantum circuit $V$ and two numbers $r,s\in[0,1].$} and polynomials $p_1,p_2,p_3$
such that:
\begin{itemize}
  \item $\forall x \in L\;\; \exists \rho \; \; \Pr_{V,r,s}\left(\left|\Tr\left(\Pi^{\ket1} V\rho V^{\dagger}\right)-r\right|\le s\right)=1$ \\
  \emph{(i.e., there exists a witness such that with probability 1 the super-verifier outputs $V$ which accepts the witness with probability which is close to $r$)},
  \item $\forall x \notin L \;\; \forall \rho \; \; \Pr_{V,r,s}\left(\left|\Tr\left(\Pi^{\ket1}V\rho V^{\dagger}\right)-r\right|\le s+p_3(1/| x|)\right)\le 1-p_2(1/|x|)$ \\
  \emph{(i.e., for any witness, with some non-negligible probability, the super-verifier outputs a circuit $V$ that accepts the witness with probability which is not close to $r$)}
\end{itemize}
where probabilities are taken over the outputs $V,r,s$ of the super-verifier and $\rho$ is a density matrix over $p_1(\lvert x\rvert)$ qubits.
\end{definition}

\begin{definition}$\left(\SuperQMA\left(m,\epsilon\right)\right)$.\label{def:SQMA}
A promise problem $ A $ is in $\SuperQMA\left(m, \epsilon\right)$ if there exists a super-verifier $V = \{(V_{x,i}, r_{x,i}, s_{x,i})\}_{i \in [m]}$
 such that:
\begin{itemize}
    \item $ \forall x \in A_{\text{yes}}, \exists \rho $: 
    $ 
    \Pr_{i} \left( \left| \Tr(\Pi^{(1)} V_{x,i} \rho V^{\dagger}_{x,i}) - r_{x,i} \right| \leq s_{x,i} \right) = 1.
    $
    \item $ \forall x \in A_{\text{no}}, \forall \rho $: 
    $ 
    \Pr_{i} \left( \left| \Tr(\Pi^{(1)} V_{x,i} \rho V^{\dagger}_{x,i}) - r_{x,i} \right| \leq s_{x,i} + \epsilon \right) \leq 1 - 1/m,
    $
    
\end{itemize}
where probabilities are taken over $i\in [m]$, $\rho$ is a density matrix on $p(|x|)$ qubits, $r_{x,i}, s_{x,i}\in[0,1]$ and $1/\poly(n)\le\epsilon\le1$. We additionally assume there exists a classical algorithm which, given any $i\in[m]$, efficiently computes  $(V_{x,i},r_{x,i},s_{x,i})$ in time polynomial in the number of qubits.
\end{definition}
\noindent Note our definition allows \emph{exponentially} many checks, so long as each check can be efficiently generated on demand. We further remark that our definition $\SuperQMAPoly$, i.e., with $m=\poly(|x|)$, coincide with $\QMAplus$ from Ref.~\cite{aharonovLatticeProblemQuantum2003}. However, to the best of our knowledge, our $\SuperQMAExp$ with $m=\exp(|x|)$ has not been considered elsewhere before.

\begin{definition}\label{def:PSQMA}
    $(\PSQMA(m,\epsilon))$. We define it exactly as \cref{def:SQMA} but with the promise that $\rho=\rho_A\otimes\rho_B$ where $\rho_A$ and $\rho_B$ are density matrices on $p_1(|x|)$ and $p_2(|x|)$ qubits, respectively. 
\end{definition}
\begin{definition}[$\mathsf{Q\Sigma_i}$  \cite{gharibianQuantumGeneralizationsPolynomial2018}]\label{def:QSigma}
A promise problem $A=(A_{\mathrm{yes}},A_{\mathrm{no}})$ is in $\mathsf{Q\Sigma_i}(c,s)$ for
polynomial-time computable functions $c,s:\mathbb{N}\to[0,1]$ if there exists a polynomially
bounded function $p:\mathbb{N}\to\mathbb{N}$ and a polynomial-time uniform family of quantum
circuits $\{V_n\}_{n\in\mathbb{N}}$ such that for every $n$-bit input $x$, $V_n$ takes $p(n)$-qubit density operators $\rho_1,\ldots,\rho_i$ as quantum proofs and outputs a single
qubit, then:
\begin{itemize}
  \item \textbf{Completeness:} If $x\in A_{\mathrm{yes}}$, then
  $\exists \rho_1\,\forall \rho_2\,\cdots\,Q_i\rho_i$ such that $V_n$ accepts
  $(\rho_1\otimes \rho_2 \otimes \cdots \otimes \rho_i)$ with probability $\ge c$.

  \item \textbf{Soundness:} If $x\in A_{\mathrm{no}}$, then
  $\forall \rho_1\,\exists \rho_2\,\cdots\,\overline{Q}_i\rho_i$ such that $V_n$ accepts
  $(\rho_1\otimes \rho_2 \otimes \cdots \otimes \rho_i)$ with probability $\le s$.
\end{itemize}
Here, $Q_i$ equals $\forall$ when $i$ is even and equals $\exists$ otherwise, and
$\overline{Q}_i$ is the complementary quantifier to $Q_i$.
\begin{equation*}
\text{Define}\quad
\mathsf{Q\Sigma_i}
= \bigcup_{\,c-s\,\in\,\Omega(1/\mathrm{poly}(n))}\;
\mathsf{Q\Sigma_i}(c,s).
\end{equation*}

\end{definition}
\begin{definition}[Quantum polynomial hierarchy ($\cfont{QPH}$) \cite{gharibianQuantumGeneralizationsPolynomial2018}]\label{def:QPH}
  $\mathsf{QPH}
  \;=\;\bigcup_{i\in\mathbb{N}} \mathsf{Q\Sigma_i}$ .
\end{definition}

\begin{definition}[$\qc(c,s)$]\label{def:qc-cs}
A promise problem $A=(A_{\mathrm{yes}},A_{\mathrm{no}})$ is in $\qc(c,s)$ if there is a
polynomial-time generated quantum verifier $V_x$ which, on input $x\in\{0,1\}^n$,
receives a polynomial-size quantum proof $\rho$ and a polynomial-size classical proof
$y$, and satisfies:
\begin{itemize}
    \item \textbf{Completeness}: If $x \in A_{\text{yes}}$, then $\exists\; \rho$ such that $\forall \ket{y}$, 
    $\Pr[V_x(\rho,\ket{y})=1]\geq c$. 
    \item \textbf{Soundness}: If $x \in A_{\text{no}}$, then $\forall \rho$, $\exists \ket{y}$ such that $\Pr[V_x(\rho,\ket{y})=1]\leq s$.   
\end{itemize}  
\end{definition}
\begin{definition}[$\qc$]\label{def:qc}
We define $\qc := \qc(2/3,1/3)$. In \Cref{lem:qc-amplification}, we show that this constant-gap definition is
equivalent to the inverse-polynomial-gap definition $\bigcup_{c-s\ge 1/\poly(n)} \qc(c,s)$.
\end{definition}
\begin{definition}[$\Pqc$]\label{def:P-qc-sigma2}
$\Pqc:=\bigcup_{c-s\ge 1/\poly(n)} \Pqc(c,s)$, where $\Pqc(c,s)$ is defined as \cref{def:qc-cs} with the promise that $\rho$ is a product state $\rho_A\otimes\rho_B$.
\end{definition}

\paragraph{Observables.} In quantum mechanics, an observable is represented by a Hermitian operator $O$. Its real eigenvalues correspond to the possible outcomes of a measurement. Throughout this work, we assume without loss of generality that all observables $O_i$ are normalized such that their operator norm $\|O_i\| \leq 1$. This implies that all eigenvalues $\lambda$ of $O_i$ lie in the interval $[-1, 1]$. This is a standard normalization, as any observable $O$ can be efficiently rescaled by a factor $C \geq \|O\|$ to satisfy this condition, which correspondingly rescales the expectation values and promise gap parameters in our problems.

\begin{remark}[Range convention for recovery values]\label{rem:range-convention}
Although the physical expectation value of any normalized observable lies
in $[-1,1]$, the raw estimator produced by a shadow recovery procedure
need not lie in this interval. In our protocol-specific CSV problems, we
therefore impose as a promise that the final reported
recovery values lie in $[-1,1]$; an out-of-range reported value is treated
as an invalid prediction.
\end{remark}

\paragraph{Succinct access assumption.} When we say that we assume succinct access to a set $\{A_i^{(n)}\}_{i=1}^m$, with $n$ the natural size parameter of the instance (e.g., 
number of qubits for observables or precision parameter for a real value), we mean that given an index $i$, a $\poly(n)$-bit description of $A_i$ can be produced in $\poly(n)$-time.

\paragraph{Huang-Kueng-Preskill classical shadow framework.} \label{par:HKP}Here we briefly describe a classical shadow protocol proposed by Huang, Kueng and Preskill in Ref. \cite{HKP20}.
For an unknown $n$-qubit state $\rho$ fix an ensemble $\mathcal U$ of unitaries on $n$ qubits. In each round do the following: sample $U\sim\mathcal U$, measure $U\rho U^\dagger$ in the computational basis to get a bitstring $b_\ell\in\{0,1\}^n$, and store a succinct classical description of $U_\ell^\dagger\ket{b_\ell}\bra{b_\ell}U_\ell$. The average channel
\[
\mathcal M(\rho)=\E_{U\sim\mathcal U}\sum_{b}\braket{b|U\rho U^\dagger|b}\,U^\dagger\ket{b}\bra{b}U
\]
is invertible for tomographically complete $\mathcal U$, so a single-shot \emph{snapshot} is
\[
\hat\rho_\ell=\mathcal M^{-1}\big(U_\ell^\dagger\ket{b_\ell}\bra{b_\ell}U_\ell\big),\qquad \E[\hat\rho_\ell]=\rho
.
\]
For any observable $O$ we use $\hat o_\ell(O):=\Tr[O\,\hat\rho_\ell]$. Partition the $L$ rounds into $K$ blocks $B_1,\ldots,B_K$ of (nearly) equal size, set
\[
\overline o_k=\frac{1}{|B_k|}\sum_{\ell\in B_k}\hat o_\ell(O),
\quad
A(S,O)=\operatorname{median}\{\overline o_1,\dots,\overline o_K\}\; \text{(rounded to $\chi$ bits).}
\]
Since $\E[U_\ell^\dagger\ket{b_\ell}\bra{b_\ell}U_\ell]=\mathcal M(\rho)$, linearity gives
$\E[\hat\rho_\ell]=\rho$ and thus $\E[\hat o_\ell(O)]=\Tr(O\rho)$ (unbiased). 
The median of means provides robustness.  We need $L=O \left( \frac{\log(M)}{\epsilon^2} \max_{1 \leq i \leq M} \left\| O_i - \frac{\Tr(O_i)}{2^n} I \right\|_{\text{shadow}}^2 \right)$ samples to estimate $M$ observables up to error $\epsilon$.
\smallskip

\paragraph{Robust classical shadow version.} \label{par:HKPRobust} The variant of the HKP scheme presented in Ref. \cite{Robust} -- which has attractive features concerning the quantum ensemble $\mathcal U$ of unitaries implemented -- shares most of the above
properties from the perspective of the present work. There, the ensemble of unitaries $\mathcal U$ is taken to be as the set of unitaries of the form
\[
U = U_C U_S U_H,
\]
where $U_C$ is an $n$-qubit unitary that consists of i.i.d.\ random controlled-$Z$ gates, $U_S$ is a single layer of gates drawn i.i.d.\ from $\{I_2,S\}$,
where $S:= (I_2+iZ)/\sqrt{2}$ is a single qubit Clifford $S$ gate, and $U_H$ is a fixed single layer of single qubit Hadamard gates. This set of unitaries is a not quite tomographically
complete subset of Clifford circuits. But one can still use single-shot snapshots
\[
\hat\rho_\ell=\mathcal M^{-1}\big(U_\ell^\dagger\ket{b_\ell}\bra{b_\ell}U_\ell\big),\qquad \E[\hat\rho_\ell]=\rho
\]
such that for any observable $O$ that is not supported on the main diagonal in the computational basis one can use $\hat o_\ell(O)=\Tr[O\,\hat\rho_\ell]$
in an unbiased recovery. This is usually satisfactory, as there are other ways of estimating the main diagonal elements of $\rho$.
Hence, the framework also applies here.

\smallskip
\noindent\textbf{Local-Clifford (random Pauli).}
Here $\mathcal U=\Cl(2)^{\otimes n}$. Per round, sample independent single-qubit Cliffords $U_{j,\ell}$ (equivalently, pick Pauli bases $P_{j,\ell}\in\{X,Y,Z\}$) and measure to get bits $x_{j,\ell}$. Store the \emph{measurement record}
\[
s_\ell=\big((P_{1,\ell},b_{1,\ell}),\ldots,(P_{n,\ell},b_{n,\ell})\big),\quad b_{j,\ell}=(-1)^{x_{j,\ell}}\in\{\pm1\}.
\]
The average channel factorizes as $\mathcal M=D_{1/3}^{\otimes n}$ with inverse $D_{1/3}^{-1}(A)=3A-\Tr(A)\,I_2$, hence the snapshot factorizes sitewise as
\[
\hat\rho_\ell=\bigotimes_{j=1}^n\hat\eta_{j,\ell},\quad 
\hat\eta_{j,\ell}=3\,\ket{\psi_{j,\ell}}\bra{\psi_{j,\ell}}-I_2,
\]
where $\ket{\psi_{j,\ell}}:=U_{j,\ell}^\dagger\ket{x_{j,\ell}}$ an eigenvector of $X$ or $Y$ or $Z$. To estimate $M$ $k$-local observables up to error $\epsilon$, it suffices to take
$L=\mathcal O\big(\frac{\log M}{\epsilon^2}\max_i 4^{k}\|O_i\|_{\infty}^2\big)$ rounds. 

\smallskip
\noindent\textbf{Global Clifford.}
Here $\mathcal U=\Cl(2^n)$. Per round, sample $U_\ell\in\Cl(2^n)$ uniformly at random and measure to get $b_\ell$. Store the measurement record as
\[
s_\ell=(\cfont{Stab}_{\ell},b_{\ell})
\]
where $\cfont{Stab}$ is the efficient classical representation of the global Clifford via the stabilizer formalism and $b\in \{0,1\}^n$ the measurement outcome of that round.\\
The average channel is the global depolarizing map
\[
\mathcal M(\rho)=D_{1/(2^n+1)}(\rho),\qquad \mathcal M^{-1}(A)=(2^n+1)A-\Tr(A)I_{2^n},
\]
so the snapshot is
\[
\hat\rho_\ell=(2^n+1)\ket{\psi_\ell}\bra{\psi_\ell}-I_{2^n}.
\]
To estimate $M$ linear observables, one needs
$L=\mathcal O\big(\frac{\log M}{\epsilon^2}\max_i\Tr(O_i^2)\big)$ rounds.

\paragraph{Mao–Yi–Zhu classical shadow framework.} \label{par:MYZ}
Mao, Yi and Zhu \cite{MYZ25} extend the HKP classical-shadow protocol to $n$ qudits of odd prime local dimension $d$. Let $\F_d$ be the finite field with $d$ elements and $\omega=e^{2\pi i/d}$ a primitive $d$-th root of unity. Fix an ensemble $\mathcal E$ of unitaries on $n$ qudits. In each round: sample $U\sim\mathcal E$, measure $U\rho U^\dagger$ in the computational basis to get an outcome $b_\ell\in\F_d^{\,n}$, and store a succinct classical description of $U_\ell^\dagger\ket{b_\ell}\bra{b_\ell}U_\ell$. The average channel
\[
\mathcal M(\rho)=\E_{U\sim\mathcal E}\sum_{b}\braket{b|U\rho U^\dagger|b}\,U^\dagger\ket{b}\bra{b}U
\]
is invertible for tomographically complete $\mathcal E$, so a single-shot \emph{snapshot} is
\[
\hat\rho_\ell=\mathcal M^{-1}\big(U_\ell^\dagger\ket{b_\ell}\bra{b_\ell}U_\ell\big),\qquad \E[\hat\rho_\ell]=\rho.
\]
For any observable $O$ we use $\hat o_\ell(O):=\Tr[O\,\hat\rho_\ell]$. Partition the $L$ rounds into $K$ blocks and take the median of block-means (rounded to $\chi$ bits) as before. Since $\E[U_\ell^\dagger\ket{b_\ell}\bra{b_\ell}U_\ell]=\mathcal M(\rho)$, linearity gives $\E[\hat o_\ell(O)]=\Tr(O\rho)$ (unbiased). Again $L=O \left( \frac{\log(M)}{\epsilon^2} \max_{1 \leq i \leq M} \left\| O_i - \frac{\Tr(O_i)}{d^n} I \right\|_{\text{shadow}}^2 \right)$ samples, suffice to estimate $M$ observables up to error $\epsilon$.

\smallskip
\noindent\textbf{Local-Clifford.}
Here $\mathcal E=\Cl(d)^{\otimes n}$. Per round, sample independent single-qudit Cliffords $U_{j,\ell}$ (equivalently, pick on each site one of the $d+1$ stabilizer bases and measure there). Store the \emph{measurement record} $s_\ell=\big((\mu_{1,\ell},b_{1,\ell}),\ldots,(\mu_{n,\ell},b_{n,\ell})\big)$,
where $\mu_{j,\ell}\in\F_d\cup\{\infty\}$ labels the basis ($\mu=\infty$: $Z$-eigenbasis; $\mu=t\in\F_d$: eigenbasis of $ Z^tX$) and $b_{j,\ell}\in\F_d$ is the outcome label.
The average channel factorizes as $\mathcal M=D_{1/(d+1)}^{\otimes n}$ with inverse $D_{1/(d+1)}^{-1}(A)=(d+1)A-\Tr(A)\,I_d$, hence the snapshot factorizes sitewise:
\[
\hat\rho_\ell=\bigotimes_{j=1}^n\hat\eta_{j,\ell},\qquad 
\hat\eta_{j,\ell}=(d+1)\,\ket{\phi_{\mu_{j,\ell},b_{j,\ell}}}\bra{\phi_{\mu_{j,\ell},b_{j,\ell}}}-I_d,
\]
where $\ket{\phi_{\mu,a}}$ is the eigenvector in the chosen stabilizer basis. 
To estimate $M$ $k$-local observables up to error $\epsilon$, it suffices to take
$L=\mathcal O\Big(\frac{\log M}{\epsilon^2}\max_i d^{2k}\,\|O_i\|_{\infty}^2\Big)$.

\smallskip
\noindent\textbf{Global Clifford.}
Here $\mathcal E=\Cl(d^n)$. Per round, sample $U_\ell\in\Cl(d^n)$ uniformly at random and measure to get $b_\ell\in\F_d^{\,n}$. Store the measurement record: $s_\ell=(\cfont{Stab}_{\ell},b_{\ell})$, where $\cfont{Stab}$ is the efficient classical representation (via stabilizer formalism) of the global Clifford and $b$ the $d$-ary outcome string. The average channel is the global depolarizing map
\[
\mathcal M(\rho)=D_{1/(d^n+1)}(\rho),\qquad \mathcal M^{-1}(A)=(d^n+1)A-\Tr(A)I_{d^n},
\]
so the snapshot is $\hat\rho_\ell=(d^n+1)\ket{\psi_\ell}\bra{\psi_\ell}-I_{d^n}$. To estimate $M$ linear observables, one needs
$L=\mathcal O\Big(\frac{\log M}{\epsilon^2}\max_i [(2d-3)\Tr(O_i^2)+2\|O_i\|_{\infty}^2]\Big)$ rounds.

\section{Complexity of $\ObsCon$}\label{scn:obscon}
The \cref{def:CSV} is overloaded for the general hardness results we are about to present. For that reason we will now recast it in a more abstract form. Notice that we will come back to the full-fledged definition when we consider specific classical shadow protocols.

\begin{definition}[Observable consistency ($\ObsCon$)]\label{def:obscon}
The input is a set of observables, as in \cref{def:classical-shadow}, along with their target expectation values $(O_i, y_i)_{i=1}^m$, for which we assume succinct access, and parameters $\alpha$ and $\beta$ satisfying $\beta-\alpha\geq 1/\poly(n)$. We further assume w.l.o.g. that $y_i\in [-1,1]$ and $0\le\alpha<\beta\le2$.  The output is to decide between the following cases:
\begin{itemize}
    \item \textbf{Yes}: $\exists$  $n$-qubit state $\rho$ such that $\forall$ $i\in [m]$, 
    $\left| \Tr\left(O_i\rho\right)-y_i\right|\leq \alpha$.
    \item \textbf{No}: $\forall$ $n$-qubit states $\rho$, $\exists$ $i\in [m]$ such that $\left| \Tr \left(O_i\rho\right)-y_i\right|\geq \beta$.
\end{itemize}
\end{definition}
\begin{lemma}\label{lem:CSV=ObsCon}
    $\CSV$ and $\ObsCon$ are equivalent under polynomial-time many-one reductions.
\end{lemma}
\begin{proof}
Both directions are straightforward
\begin{itemize}
\item $\CSV\le\ObsCon$. Keep the same observables and define $y_{i}:=A(S,i)$.
\item $\ObsCon\le\CSV$. Keep the same observables, use a dummy shadow $S$ and a recovery algorithm $A(S,i)$ that ignores $S$ and outputs $y_i$. 
\end{itemize}
\end{proof}
We will analyze the complexity of this problem in two regimes, distinguished by the number of observables $m$.

\subsection{Polynomially many observables ($\ObsConPoly$)}\label{sscn:obsconpoly}
\begin{definition}
    [$\ObsConPoly$] Same as \cref{def:obscon} with $m=\poly(n)$.
\end{definition}
\begin{proposition}\label{prop:obsconpoly-in-superqmapoly}
$\ObsConPoly\in \SuperQMAPoly$.
\end{proposition}
\begin{proof}
    \textbf{Verification procedure:} Given the state $\rho$ the verifier picks  $i\in[m]$ uniformly at random and measures the observable $O_i$ on the state $\rho$. This will give one of its eigenvalues $\lambda_j$. Then define a biased coin that gives heads with probability $p_h=\frac{1+\lambda_j}{2}$ and tails with probability $p_t=1-p_h$. Flip the coin and accept on heads, reject on tails.\\
    The overall acceptance probability becomes $\Pr(\text{accept}|i)=\frac{1}{2}(1+\Tr(\rho O_i))$.
    Set the target probability to be $r_{x,i}=\frac{1}{2}(1+y_i)$ and the tolerance parameter $s_{x,i}=\frac{\alpha}{2}$, 
    uniform $\forall i$. Then our protocol works with $\epsilon=\frac{\beta-\alpha}{4}$.\\
    \\
    \textbf{Completeness:} From the promise of the YES case we have that $\forall i\;|\Tr(O_i\rho)-y_i|\leq \alpha$. So we find
    \[
    \forall i\;\;|\Pr(\text{accept}|i)-r_{x,i}|=\frac{1}{2}|\Tr(O_i\rho)-y_i|\leq \frac{\alpha}{2}=s_{x,i}. 
    \]
    In other words $\Pr_i(|\Pr(\text{accept}|i)-r_{x,i}|\leq s_{x,i})=1$.\\
    \\
    \textbf{Soundness:} From the promise of the NO case we have that there exists at least one $i$, say $i^*$, s.t. $|\Tr(O_{i^*}\rho)-y_{i^*}|\geq \beta$. For $i^*$ we then have
    \[
    |\Pr(\text{accept}|i^*)-r_{x,i^*}|=\frac{1}{2} |\Tr(O_{i^*}\rho)-y_{i^*}|\geq \frac{\beta}{2}> s_{x,i^*}+\epsilon
    \]
    Where the last inequality holds since $s_{x,i^*}+\epsilon =\frac{\alpha+\beta}{4}$ and $\beta-\alpha\ge 1/\poly(n)$.\\
    \\
    In other words $\Pr_i[|\Pr(\text{accept}|i)-r_{x,i}|\leq s_{x,i}+\epsilon]\leq 1-\frac{1}{m}$.
\end{proof}
\begin{proposition}\label{prop:obsconpoly-hard-superqmapoly}
    $\ObsConPoly$ is $\SuperQMAPoly$- hard.
\end{proposition}
\begin{proof} 
For input $x$ the $\SuperQMAPoly$ super-verifier provides $m=\poly(|x|)$ checks $\{(V_i, r_i, s_i)\}_{i=1}^m$, with $r_i,s_i\in[0,1]$ and a global gap parameter $1/\poly(n)\le\epsilon\le 1$.\\
\\
\textbf{Parameter setting:}  
Define
\[
\epsilon' := \frac{\epsilon}{2}, \quad \tau := \frac{\epsilon}{4},\quad
s_i' := \max\{s_i,\tau\},\quad t_i=\frac{\tau}{s_i'}.
\]
\textbf{Mapping.}  
The reduction outputs the $\ObsConPoly$ instance $\{(O_i, y_i)\}_{i=1}^m$ with uniform thresholds $\alpha, \beta$ defined by
\[
O_i := t_i\big(V_i^\dagger \Pi^{(1)} V_i\big), \quad
y_i := t_ir_i,\quad
\alpha := \tau, \quad \beta := \tau + \tau\epsilon'.
\]
Notice that this choice of parameters gives us a $\beta-\alpha=\tau\epsilon'\ge\frac{\epsilon^2}{8}\ge\frac{1}{\poly(n)}$.\\
\\
\textbf{Completeness (YES case).}  
If the original $\SuperQMAPoly$ instance is YES, there exists a witness $\rho$ such that
\[
\forall i\;\;:\big|\Tr(V_i^\dagger \Pi^{(1)} V_i \,\rho) - r_i\big| \le s_i\le s_i'.
\]
Multiplying by $t_i$ gives
\[
\big|\Tr(O_i \rho) - y_i\big|
= t_i\big|\Tr(V_i^\dagger \Pi^{(1)} V_i \,\rho) - r_i\big|
\le t_is_i'=\tau=\alpha.
\]
Thus the same $\rho$ satisfies $\big|\Tr(O_i \rho) - y_i\big| \le \alpha$ for all $i$, so the mapped instance is a YES-instance of $\ObsConPoly$.\\
\\
\textbf{Soundness (NO case).}  
If the original $\SuperQMAPoly$ instance is NO, then for every state $\rho$ there exists some index $i^*$ with
\[
\big|\Tr(V_{i^*}^\dagger \Pi^{(1)} V_{i^*} \,\rho) - r_{i^*}\big| > s_{i^*} + \epsilon\ge s_{i^*}'-\tau+\epsilon=s_{i^*}'+\frac{3\epsilon}{4}> s_{i^*}'+\epsilon'.
\]
Multiplying by $t_{i^*}$ yields
\[
\big|\Tr(O_{i^*} \rho) - y_{i^*}\big| > t_{i^*}s_{i^*}'+t_{i^*}\epsilon'=\tau+\frac{\tau}{s_{i^*}'}\epsilon'\ge\tau+\tau\epsilon'=\beta.
\]
Thus the mapped instance violates the uniform $\beta$-threshold for the index $i^*$, matching the $\ObsConPoly$ NO condition.

\end{proof}

\begin{theorem}[\cite{aharonovLatticeProblemQuantum2003}]\label{Thm:QMAplus=QMA}
    $\QMA=\SuperQMAPoly$.
\end{theorem}

\begin{corollary}\label{cor:obsconpoly}
    $\ObsConPoly$ is $\QMA$-complete.
\end{corollary}
\begin{proof}
    This follows from \cref{prop:obsconpoly-in-superqmapoly,prop:obsconpoly-hard-superqmapoly} and \cref{Thm:QMAplus=QMA}.
\end{proof}
\begin{corollary}\label{cor:CSVpoly}
    $\CSV_{\poly}$ is $\QMA$-complete.
\end{corollary}
\begin{proof}
    This follows from \cref{cor:obsconpoly} and \cref{lem:CSV=ObsCon}.
\end{proof}
Note here that the $\QMA$-hardness result need not go through the super-verifier machinery. We can directly reduce from the $\CLDM$ problem which is known to be $\QMA$ complete under Karp reductions \cite{BG22}. We give this reduction in \cref{sec:A}. The reason we use this machinery is because it will become helpful in the $\exp$ regime that we analyze next.

\subsection{Exponentially many observables ($\ObsConExp$)}\label{sscn:obsconexp}

We now move to analyze the case where the observables can be exponentially many, albeit we have succinct access to them. Here the super-verifier machinery we developed for the $\poly$ regime will help us extract completeness results for $\ObsConExp$ immediately.
\begin{definition}
 [$\ObsConExp$]\label{def:ObsConexp} Same as \cref{def:obscon} with $m=\exp(n)$.   
\end{definition}
\begin{proposition}\label{prop:obsconexp-in-superqmaexp}
    $\ObsConExp \in \SuperQMAExp. $       
\end{proposition}
\begin{proof}
The proof follows in the same manner as in the $\poly$-case, \cref{prop:obsconpoly-in-superqmapoly}. The verifier only needs to generate and execute a single, randomly chosen check $(O_i,y_i)$. Since the $\ObsConExp$ instance guarantees that any such pair can be generated in polynomial time given the index i, the verifier remains efficient. The soundness guarantee of $1/m$ holds, where $m$ is now exponential in the number of qubits.
\end{proof}
\begin{proposition}\label{prop:obsconexp-hard-superqmaexp}
    $\ObsConExp$ is $\SuperQMAExp$-hard.
\end{proposition}
\begin{proof}
  The proof again carries over from the $\poly$ case. Here for each one of the exponentially many checks of the $\SuperQMAExp$, the mapping in \cref{prop:obsconpoly-hard-superqmapoly} gives, in polynomial time, one of the $\exp$ many pairs $(O_i,y_i)$ of $\ObsConExp$ along with the global parameters $\alpha,\beta$. This is all we need since we assume succinct access to both the checks and the pairs. 
\end{proof}

\begin{corollary}\label{cor:ObsConexp-SQMAexp_hard}
    $\ObsConExp$ is $\SuperQMAExp$-complete.
\end{corollary}
\begin{proof}
 Follows from \cref{prop:obsconexp-in-superqmaexp,prop:obsconexp-hard-superqmaexp}.   
\end{proof}

We next show that $\SuperQMAExp$ coincides with the second level of a quantum-classical variant of the quantum polynomial hierarchy ($\cfont{QPH}$). By $\cfont{QPH}$ we refer to the hierarchy $\cfont{Q\Sigma_i}$ of Ref.\ \cite{gharibianQuantumGeneralizationsPolynomial2018} (see \cref{def:QSigma}). We call our variant $\qc$ (see \cref{def:qc}).

We first prove an amplification lemma showing that the constant-gap and inverse-polynomial-gap definitions of $\qc$ coincide.
\begin{lemma}[Amplification for $\qc(c,s)$]\label{lem:qc-amplification}
Let $L\in \qc(c,s)$, where $\Delta:=c-s\ge 1/\poly(n)$. Then for any polynomial $r$, $L\in \qc(1-2^{-r(n)},2^{-r(n)})$.
\end{lemma}
\begin{proof}
The strategy is to invoke the standard weak-error amplification of the $\QMA$ protocol, see \cite{kitaev2002classical} for a clear exposition. The only obstacle here is that a priori in the NO case, it seems like there is no single poly-sized universal proof that works for all possible existential proofs. In the following we show that this is not the case.

Soundness: $\forall\rho\;\exists z$ such that $\Tr(A_z\rho)\le s$ or equivalently $\max_{\rho}\min_z\Tr(A_z\rho)\le s$. We now have
\begin{equation}\label{eq:badistrb}  
\max_{\rho}\min_z\Tr(A_z\rho)=\max_{\rho}\min_{\mu}\mathbb{E}_{z\sim\mu}\Tr(A_z\rho)=\min_{\mu}\max_{\rho}\mathbb{E}_{z\sim\mu}\Tr(A_z\rho)\le s
\end{equation}
where the first equality comes from the fact that the minimum of a linear function over the probability simplex is achieved at an extreme point and the second equality comes from Sion's minimax theorem (see \cite{Sion58}). From \cref{eq:badistrb} we conclude that there exists a ``bad" distribution $\mu^*$ such that $\max_{\rho}\mathbb{E}_{z\sim\mu^*}\Tr(A_z\rho)\le s$. For $A_{\mu^*}=\mathbb{E}_{z\sim\mu^*}A_z$ we get
\[
\max_{\rho}\Tr(A_{\mu^*}\rho)\le s\quad\text{and so}\quad \lambda_{\max}(A_{\mu^*})\le s. 
\]
Next we need to show that we can sparsify this distribution into a poly-size list of challenges. For this we use the matrix Hoeffding bound, defined as follows:\\
\\
Let $\{X_j\}$ be a sequence of independent, random, self adjoint matrices of dimension $d$ and $\{A_j\}$ a sequence of fixed self adjoint matrices. Assume $\mathbb{E}X_j=0$ and $X_j^2\le A_j^2$. Then matrix Hoeffding gives
\[
\Pr\left(\lambda_{\max}\left(\sum_jX_j\right)\ge t\right)\le d\cdot e^{\frac{-t^2}{8\sigma^2}}
\]
where $\sigma^2=\|\sum_jA_j^2\|$.\\
\\
For us $X_j:=\frac{1}{T}(A_{z_j}-A_{\mu^*})$. Then $\sum_{j=1}^TX_j=\frac{1}{T}\sum_{j=1}^TA_{z_j}-A_{\mu^*}=A_L-A_{\mu^*}$. So the matrix Hoeffding will bound $\lambda_{\max}(A_L-A_{\mu^*})$. Let us now check the assumptions of the matrix Hoeffding. Since $A_{\mu^*}, A_z$ are Hermitian then $X_j$'s are self adjoint and also independent and random since we sample from a distribution independently. We also have that 
\[
\mathbb{E}[X_j]=\frac{1}{T}(\mathbb{E}[A_{z_j}]-\mathbb{E}[A_{\mu^*}])=0
\]
Since $\mathbb{E}[A_{z_j}]=A_{\mu^*}$ by definition and $\mathbb{E}[A_{\mu^*}]=A_{\mu^*}$ since there is no randomness here. Now for $\{A_j\}$ (the fixed sequence of self-adjoint matrices) we pick $B_j$, to be defined in a bit. Since 
\[
0\preceq A_{z_j}\preceq I,\quad 0\preceq A_{\mu^*}\preceq I.
\]
The first one holding because $A_z$ is an acceptance POVM and the second one because $A_{\mu^*}$ is a convex combination of such POVM's. So we have $(A_{z_j}-A_{\mu^*})^2\preceq I$ and so
\[
X_j^2=\frac{1}{T^2}(A_{z_j}-A_{\mu^*})^2\le\frac{1}{T^2}I
\]
Now by choosing $B_j:=\frac{1}{T}I$ we satisfy the last assumption, i.e., $X_j^2\le B_j^2$. We also have $\sigma^2=\|\sum_jB_j^2\|=\|\sum_{j=1}^T\frac{1}{T^2}I\|=\frac{1}{T}$.\\
\\
So now by applying the Hoeffding matrix we get:
\[
\Pr\left[\lambda_{\max}(A_L-A_{\mu^*})\ge \eta\right]\le d\cdot e^{\frac{-T\eta^2}{8}}
\]
We can run the same analysis for $-X_j$ and so in total from a union probability bound we get 
\[
\Pr\left(\|A_L-A_{\mu^*}\|_{\infty}\ge \eta\right)\le 2d\cdot e^{\frac{-T\eta^2}{8}}
\]
We can now see that for $T>\frac{8}{\eta^2}\log(2d)=\poly(n)$ the above probability is strictly smaller than 1 and so there exists a poly sized list that gives $\|A_L-A_{\mu^*}\|\le\eta$. We can now invoke the Weyl inequality: $|\lambda_k(A+B)-\lambda_k(A)|\le\|B\|_{\infty}$. For us $A:=A_{\mu^*}, B:=A_L-A_{\mu^*}$ and so 
\[
\lambda_{\max}(A_L)\le \lambda_{\max}(A_{\mu^*})+\|A_L-A_{\mu^*}\|_{\infty}\le s+\eta .
\]
Now we apply the standard weak amplification argument. Let
\[
    \Delta:=c-s,\qquad \eta:=\Delta/4,
    \qquad
    \theta:=\frac{c+s+\eta}{2}.
\]
The universal prover sends the list $L=(z_1,\dots,z_T)$ guaranteed
above. The verifier receives from the existential prover $R=\poly(n,1/\Delta)$ quantum proof blocks.
For each block $r\in[R]$, it chooses $z\in L$ uniformly at random, runs
the original verifier with classical proof $z$ on the $r$-th block, and
records the acceptance bit $X_r$. It accepts iff $\hat A:=\frac1R\sum_{r=1}^R X_r\ge \theta$. The completeness and soundness analysis is now the standard weak amplification argument (see \cite{kitaev2002classical}).
\end{proof}

\begin{lemma}\label{lem:qc-USQMA}
    $\qc\subseteq  \SuperQMAExp$.
\end{lemma}
\begin{proof}
    Given a verifier $V$ for a language $L\in \qc $ we construct a super-verifier $V'$ for $L$. Hardwire the classical proof $c$ into the verifier $V$ of $\qc$, let us call it $V_c$. The super-verifier's checks are now parametrized by the classical strings $c\in\{0,1\}^{c(n)}$, i.e., $m=2^{c(n)}$. Construct a super-verifier $V'$ that on input $x$ picks uniformly at random a challenge $c$ and outputs the check $\left(V_c, r=1, s=1/3\right)$. This satisfies the definition of $ \SuperQMAExp $ with $\epsilon=1/6$.\\
    \textbf{Completeness:} Let $ x\in L$ then  $\exists\rho$ such that $\forall c\; \Tr(\Pi^{(1)}V_c\rho V_c^\dagger)\geq2/3$. It is easy to see that the condition $|\Tr(\Pi^{(1)}V_c\rho V_c^{\dagger})-1|\leq1/3$ is  satisfied for all $c$.\\
\textbf{Soundness:} Let $x\notin L$ then $\forall\rho\;\exists c$ s.t. $ \Tr(\Pi^{(1)}V_c\rho V_c^\dagger)\leq1/3$. This means that the condition  $|\Tr(\Pi^{(1)}V_c\rho V_c^{\dagger})-1|>s+\epsilon=1/3+1/6=1/2$ is satisfied for at least one $c$, for each $\rho$, so with probability $\geq \frac{1}{m}$.
\end{proof}
\begin{lemma}\label{lem:USQMA-qc}
  $\SuperQMAExp \subseteq \qc $ .
\end{lemma}
\begin{proof}[Proof sketch]
The $\forall$-prover names a check $i$ that violates the super-verifier condition, and the $\qc$ verifier estimates the acceptance probability of $V_i$ by running it on $k=O(n/\epsilon^2)$ proof registers and checking whether the empirical average lies within $s_i+\epsilon/2$ of $r_i$. Completeness follows by Hoeffding for honest $k$-copy witnesses, while soundness follows from the same Markov argument as in \cite{aharonovLatticeProblemQuantum2003} , which applies even when the $k$ registers are entangled. By \Cref{lem:qc-amplification}, we can amplify to the standard constant-gap definition of $\qc$.
\end{proof}

\begin{corollary}\label{cor:qc=SQMA_exp}
    $\qc=  \SuperQMAExp$.
\end{corollary}
\begin{proof}
 Follows from \cref{lem:qc-USQMA,lem:USQMA-qc}.   
\end{proof}

An important variant of $\ObsConExp$, because of its connection with a triply efficient classical shadow protocol for all the $n$-bit Pauli observables ~\cite{kingTriplyEfficientShadow2025}, is when we have a constant gap parameter:
\begin{definition}
    [$\ObsConExp^{\Theta(1)}$]Same as \cref{def:ObsConexp} with $\beta-\alpha=\Theta(1)$. 
\end{definition}
It is easy to see that even for a constant gap we still have $\SuperQMAExp$-completeness:
\begin{lemma}\label{l:obsconconstant}
  $\ObsConExp^{\Theta(1)}$ is $\SuperQMAExp$-complete.  
\end{lemma}
\begin{proof}
Containment follows exactly as in \cref{prop:obsconexp-in-superqmaexp}. For hardness, use
$\SuperQMAExp=\qc$ (\cref{cor:qc=SQMA_exp}), and the amplification theorem for $\qc$
(\cref{lem:qc-amplification}), so that any
$L\in\SuperQMAExp$ has a $\qc$ verifier with completeness/soundness
$2/3,1/3$. For each classical challenge $c$, let $M_c$ be the induced
acceptance POVM of the amplified verifier with $c$ hardwired, and
output the $\ObsConExp^{\Theta(1)}$ instance $\{O_c:=M_c,\; y_c:=1,\; \alpha:=1/3,\; \beta:=2/3\}$
In the YES case, some $\rho$ satisfies $\Tr(M_c\rho)\ge2/3$ for all $c$,
so $|\Tr(O_c\rho)-1|\le1/3$. In the NO case, for every $\rho$ some $c$
satisfies $\Tr(M_c\rho)\le1/3$, so $|\Tr(O_c\rho)-1|\ge2/3$. Thus the
constructed instance has constant gap.
\end{proof}

We conclude this section by showing an easy lower and upper bound for this new class.

\begin{proposition}\label{prop:qc}
    $\QMA\subseteq \qc\subseteq \QS\subseteq \PSPACE$.
\end{proposition}

\begin{proof}
    The first inclusion follows since the 
    $\qc$ verifier can simply ignore the proof from the $\forall$ 
    prover and run the $\QMA$ verifier. The second follows since the verifier of $\QS$ can measure the $\forall$ proof in the computational basis, essentially rendering the quantum proof to a classical one, or rather a distribution of classical ones, and then simulate the verifier 
    of $\qc$. As for the third inclusion it 
    is proven in 
    Ref.\ \cite{gharibianQuantumGeneralizationsPolynomial2018}. 
    The proof was based on the observation that $\QS=\cfont{QRG(1)}$, where $\cfont  {QRG(1)}$ and its containment in $\PSPACE$ are presented in Ref.\ \cite{JW08}. 
\end{proof}

\section{Complexity of specific protocol classical shadows}\label{scn:HKPhardness}

The QMA-completeness of $\ObsConPoly$, and so of $\CSV_{\poly}$ (the more involved definition of the problem that will come in handy on this section (\cref{def:CSV})), demonstrates the problem's fundamental difficulty. We now further explore the complexity of this problem by casting it on specific, structured measurement protocols. We show that the hardness persists for two such protocols, namely the HKP with a local Clifford ensemble protocol, given in Ref.\ \cite{HKP20} and the MYZ which is its qudit generalization, given in
Ref.\ \cite{MYZ25}. Additionally we give an efficient algorithm result for the HKP, MYZ protocols with global Clifford ensemble.
\begin{remark}\label{rem:CSVp-to-ObsConp}
    For a fixed shadow protocol $P$, the reduction $\CSV_P\le \ObsCon_P$ is immediate by keeping the set of observables the same and setting $y_i:=A(S,i)$; the converse direction is protocol dependent and need not be trivial.
\end{remark}
\subsection{HKP classical shadows}

First we focus on the Huang, Kueng and Preskill protocol using local Clifford measurements \cite{HKP20} (for details on the protocol check \cref{par:HKP}). We will call the $\CSV$ problem that is based on this protocol $\CSV_{\text{HKP}}$.

\begin{definition}[$\CSV_{\cfont{HKP}}$]\label{def:hkpcsv} The definition is the same as \cref{def:CSV} only now our classical shadow has the structure dictated by the HKP local Clifford measurement protocol. That means the following:
    \begin{itemize}
        \item The shadow $S$ consists of $L=\poly(n)$ strings, $\{s_i\}_{i=1}^L$, each of which encodes the Pauli-basis measurement and the measurement outcome of each round of the protocol. More formally, each string will be of the form $(P,b)_1,\dots,(P,b)_n$ with $P\in\{X,Y,Z\}$ and $b\in\{-1,1\}$, where $(P,b)_i$ denotes the basis and the measurement outcome of the $i$-th qubit.
        \item $O$ is a set of $m=\poly(n)$ $k$-local observables on $n$-qubits, with $k=O(1)$.
        \item The recovery algorithm applies the inverse channel to extract the snapshot operators $\hat{\eta}$ from $S$ and aggregates estimates via the \emph{Median of Means} (MoM) technique, subject to the range convention of \Cref{rem:range-convention}.
    \end{itemize}
    
\end{definition}
We sometimes speak of the snapshot operator, $\hat{\eta}$, associated to a stored string $(P,b)$; it is not stored explicitly but computed in recovery. The map $s\rightarrow \hat{\eta}$ is a bijection onto the set of achievable snapshots, so storing strings or storing
snapshots are equivalent representations, hence we freely use “strings” and “snapshots” interchangeably when no confusion can arise.

\begin{definition}[$\iDLH$]
  Given a local Hamiltonian on a chain of $n$ qu-$d$-its $H = \sum_{i=1}^{n-1} H_{i,i+1}$ and thresholds $\alpha,\beta$ with $\beta-\alpha\ge 1/\poly(n)$, decide
  \begin{itemize}
    \item YES: $\lmin(H) \le \alpha$.
    \item NO: $\lmin(H) \ge \beta$.
  \end{itemize}
\end{definition}
\begin{definition}[$\iDCLDM$]
  Given local density matrices $\sigma_{i,i+1}$ for $i\in[n-1]$ for a system of $n$ qu-$d$-its and parameters $\alpha, \beta$ with $\alpha\le\negl (n),\;\beta-\alpha\ge1/\poly(n)$, decide
  \begin{itemize}
    \item YES: $\exists\rho\in\density(d^n)\;\forall i\in[n-1]\colon\norm{\Tr_{\overline{i,i+1}}(\rho) - \sigma_{i,i+1}}_{\Tr} \le \alpha$.
    \item NO: $\forall\rho\in\density(d^n)\;\exists i\in[n-1]\colon\norm{\Tr_{\overline{i,i+1}}(\rho) - \sigma_{i,i+1}}_{\Tr} \ge \beta$.
  \end{itemize}
\end{definition}
\begin{theorem}\label{Thm:iDCLDM=QMA}
    $\iDCLDM$ on 8-level qudits is $\QMA$-complete.
\end{theorem}
\begin{proof}[Proof sketch]
The high level idea is to combine the results of Ref.\ \cite{BG22}, where they prove that the $\CLDM$ problem is $\QMA$-complete under Karp reductions via the machinery of simulatable codes and of Ref.\ \cite{HNN13} where they show that $\iDLH$ on a chain of 8-level qudits is still $\QMA$-complete. For details, see \cref{sec:B}.
\end{proof}
\begin{theorem}\label{thm:1DCLDM-CSV}
  $\iDCLDM_{d=2^\ell}\le\CSV_{\text{HKP}}$.
\end{theorem}
\begin{proof}
  Here we assume qudits of dimension $d=2^\ell \in O(1)$, so that we can treat each qudit as $\ell$ qubits.
  We use the HKP shadow protocol with an ensemble of local Clifford operators, so that we end up applying random Pauli measurements,  i.e., the product of $\ell n$ single-qubit Paulis.
  Let $\sigma_{1,2},\dots,\sigma_{n-1,n}\in\density(d^2)$ be the input density matrices. We now describe the reduction from the $\sigma_{i,i+1}$ to a shadow.
  Let $\set{\hat\eta_j}_{j\in[m]}$ be the set of all $m=6^\ell$ possible snapshots on $\ell$ qubits. 
  For clarity, we are measuring each qubit in one of Pauli $X$, $Y$, or $Z$ uniformly at random, therefore the set of possible snapshots on a single qubit are of form $3\ket{\psi}\bra{\psi}-I$ for $\ket{\psi}$ an eigenvector of $X$, $Y$, or $Z$. 
  In turn, a snapshot on a given qu\emph{d}it is a tensor product of $\ell$ such terms.\\
  Suppose that there exists a state $\rho\in\density(d^n)$ whose local
marginals $\tau_{i,i+1}:=\Tr_{\overline{i,i+1}}(\rho)$ satisfy
$\|\tau_{i,i+1}-\sigma_{i,i+1}\|_{\Tr}\le \alpha_{\CLDM}$.
Since $\Exp[\hat{\rho}]=\rho$, we have
\begin{align}\label{eqn:local}
    \tau_{i,i+1}
    = \Tr_{\overline{i,i+1}}(\rho)
    = \Tr_{\overline{i,i+1}}(\Exp[\hat{\rho}])
    = \sum_jp_j\Tr_{\overline{i,i+1}}(\hat\rho^{(j)})
    = \Exp[\hat\rho_{i,i+1}],
\end{align}
  where $\set{\hat\rho^{(j)}}_j$ is a collection of all possible snapshots and $p_j$ is the probability of obtaining that snapshot from the shadow protocol. Note that with our choice of shadow protocol, we have $\hat\rho = \hat\rho_{1}\otimes\dotsm\otimes\hat\rho_n$, where each $\hat\rho_i$ is a local snapshot on $\ell$-qubits.
 Thus, by \Cref{eqn:local}, the ideal local snapshot distribution reconstructs
$\tau_{i,i+1}$ exactly, and hence reconstructs the target
$\sigma_{i,i+1}$ up to the original $\CLDM$ completeness error.
Equivalently, for each edge there are probabilities $\{p_{i,j,k}\}$ such that
$\sum_{j,k}p_{i,j,k}\hat\eta_j\otimes\hat\eta_k$ is within
$\alpha_{\CLDM}$ of $\sigma_{i,i+1}$ in trace norm, for $i$ the left qudit index in a neighboring pair of qudits, and $j$ and $k$ indexing the possible snapshots on the left and right qudit, respectively.
  
  The preceding discussion shows that, at the level of the ideal HKP snapshot distribution, a globally consistent state induces compatible local snapshot distributions on every neighboring pair. Since the reduction must output a finite classical shadow, we work directly with integer counts rather than real probabilities. Let $L=\poly(n)$ be chosen sufficiently large, and let $n_{i,j,k}$ denote the number of times the two-qudit snapshot $\hat\eta_j\otimes\hat\eta_k$ appears on edge $(i,i+1)$. We impose exact overlap-count constraints, which will allow the local shadows to be stitched into global strings, and an approximate reconstruction constraint with tolerance $\epsilon$. The tolerance $\epsilon$ is chosen to absorb this $\CLDM$ error together with the finite-count approximation needed to represent the ideal snapshot
distribution by $L$ integer counts.
  \begin{subequations}\label{eq:marginal-system2}
  \begin{align}
    \norm*{\sigma_{i,i+1} - \frac1L\sum_{j,k\in[m]} n_{i,j,k}\hat\eta_j\otimes\hat\eta_k}_{\Tr} &\le \epsilon &&\forall i\in[n-1],\\
    \sum_{j\in[m]}n_{ijt} &= \sum_{k\in[m]}n_{i+1,tk} && \forall i\in[n-2]\;\forall t\in[m],\label{eq:marginal-system2:consistency}\\
    n_{i,j,k} &\in \mathbb Z_{\ge0} &&\forall i\in[n-1]\;\forall j\in[m]\;\forall k\in[m],\label{eq:marginal-system2:ge0}\\
    \sum_{j,k\in[m]} n_{i,j,k} &= L&&\forall i\in[n-1].\label{eq:marginal-system2:sum1}
  \end{align}
  \end{subequations}
  If the $\iDCLDM$ instance is YES, then for sufficiently large $L=\poly(n)$ the system in \Cref{eq:marginal-system2} is feasible. Indeed, drawing $L$ HKP snapshots from a consistent state $\rho$ gives edge counts satisfying the overlap constraints exactly, and by the HKP concentration bound the corresponding empirical edge averages are within the allowed tolerance. If \cref{eq:marginal-system2} is unsatisfiable then the reduction outputs a trivial NO-instance.
  \begin{proposition}
    There is a \emph{dynamic programming algorithm} that efficiently solves the integer program defined in \cref{eq:marginal-system2}.
\end{proposition}
\begin{proof}
 \textbf{Goal:} To determine if there exists a sequence of $N_1,N_2,...,N_{n-1}$ such that:
    \begin{itemize}
        \item Each $N_i$ is an $m\times m$ matrix where:
        \begin{itemize}
            \item Its elements are non-negative integers.
            \item The sum of all its elements is $L$.
        \end{itemize}
        \item $\|\sigma_{i,i+1}-\frac{1}{L}\sum_{j,k}(N_i)_{j,k}\;\hat{\eta}_j\otimes\hat{\eta}_k\|_{\Tr}\leq\epsilon$ 
        \item $M_r(N_i)_t=M_l(N_{i+1})_t$
    \end{itemize}
where $M_r(N_i)_t=(c_1,...,c_m)$, with $c_k$ being the sum of the elements of the k-th column of $N_i$, is the right marginal of $N_i$ and $M_l(N_{i+1})_t=(r_1,...,r_m)$, with $r_k$ being the sum of the elements of the k-th row of $N_{i+1}$, is the left marginal of $N_{i+1}$. So this relation makes sure that the number of times that each snapshot type appears on the right marginal of $\sigma_{i,i+1}$ matches the number of times the same snapshot type appears on the left marginal of $\sigma_{i+1,i+2}$.   \\
\\
\textbf{Domain:}
$D=\{N\in \mathbb{Z}_{\geq0}^{m\times m} | \sum_{j,k}N_{j,k}=L   \}$,  $|D|=\binom{L+m^2-1}{m^2-1}=O(L^{m^2-1})$\\
The size of the domain, i.e., the number of different $m\times m$ matrices whose elements sum to $L$ is given by a, standard in combinatorics, ``balls-and-bars" theorem \cite{Tucker06}.\\
\\
\textbf{Trace-norm filter:}
$U_i=\{N\in D:\|\sigma_{i,i+1}-\frac{1}{L}\sum_{j,k}N_{j,k}\;\eta_j\otimes\eta_k\|_{\Tr}\leq\epsilon     \}$.\\
\\
\textbf{Cost:} (Precompute everything). Calculating the trace norm of this $2^{2\ell}\times 2^{2\ell}$ matrix takes $O((2^{2\ell})^3)=O(1)$ time. 
(This is the cost of the singular value decomposition step \cite{GR70}.) We need to do that $\forall\; N\in D$ so for each link $i$ the cost of this step is $O(|D|)$.\\
\\
\textbf{Marginal-match relation:} $R=\{(N,N')\in D\times D\;|\;\forall\;k,\; \sum_jN_{j,k}=\sum_jN'_{k,j}   \}$.\\
\\
\textbf{Cost:} (Compute as we go). For every $N\in U_i$ we check all $N'\in D$. One of these checks takes $O(m^2)=O(1)$-time, so to check everything for the current $N$ would take $O(|D|)$-time and to check everything at the current link takes $O(|D|^2)=O(L^{2m^2-2})$-time (note that $|U_i|\leq|D|$).\\
\\
So now we have a classical \emph{constraint satisfiability problem} on a path $x_1-x_2-\cdots -x_{n-1}$ where each variable $x_i\in D$ with:
\begin{itemize}
    \item Trace constraint: $x_i\in U_i$.
    \item Marginal constraint: $(x_i, x_{i+1})\in R$.
\end{itemize}
This can be solved via the simple \cref{alg:sequence_check}.

\begin{algorithm}
    \caption{Global sequence existence check}
    \label{alg:sequence_check}
\begin{algorithmic}[1]
    \State $F_1 \gets U_1$
    \For{$i \gets 1 \text{ to } n-2$} 
        \State $F_{i+1} \gets \emptyset$
        \For{$\text{each } N \in F_i$} 
            \For{$\text{each } N' \in U_{i+1}$}
                \If{$(N, N')\in R$}
                    \State $F_{i+1} \gets F_{i+1} \cup \{N'\}$ 
                \EndIf
            \EndFor
        \EndFor
        \If{$F_{i+1} = \emptyset$}
            \State \text{\textbf{reject} (``NO solution")} %
        \EndIf
    \EndFor
    \If{$F_{n-1} \neq \emptyset$}
        \State \text{\textbf{accept} (``YES, a global sequence exists")} 
    \EndIf
\end{algorithmic}
\end{algorithm}
\noindent Notice that we can easily retrieve an accepting sequence via standard back tracking.\\
\\
\textbf{Runtime:} As we mentioned before, computing the set $U_i$ takes $O(|D|)$-time while computing the marginal relation $R$ takes $O(|D|^2)$-time. Since we need to do that for $O(n)$ links in the chain the total runtime is:
\[
T=O(n|D|^2+n|D|)=O(nL^{2m^2-2}).
\]
And since $m$ is constant and $L=\poly(n)$ the total runtime is polynomial in the size of the input. 
\end{proof}
  We define ``local shadows'' $S_{i} = \{s_{il}\}_{l\in[L]}$ by taking $n_{i,j,k}$ copies of $\hat\eta_j\otimes\hat\eta_k$.
  We can now compute permutations $f_i\in \textup S_L$, such that $\Tr_{1}(s_{il})=\Tr_2(s_{i+1,f_i(l)})$ via a perfect matching.
  Finally, we assemble the local shadows to a global shadow
  \begin{equation}
    S = \{s_{l}\}_{l\in [L]}, \quad s_l = s_{1,l}\otimes \Tr_A(s_{2,f_1(l)}) \otimes \Tr_A(s_{3,f_2(f_1(l))})\otimes \dotsm\otimes \Tr_A(s_{n-1,(f_{n-2}\circ\dotsm\circ f_1)(l)}).
  \end{equation}
  By construction, we have 
  \begin{equation}\label{constructionguarantee}
  \norm*{\Tr_{\overline{i,i+1}} \left(\frac{1}{L}\sum_{l\in[L]}s_l\right) - \sigma_{i,i+1}}_{\Tr} \le \epsilon.
  \end{equation}

  For the $\CSV_{\cfont{HKP}}$ instance, we use $K$ identical copies of the global shadow $S$ (or to be precise the string equivalent of the snapshots), which will serve as the buckets for the $\cfont{MoM}$ aggregation, essentially rendering this to an empirical average. As for the set of observables $O$, for each neighboring qudit pair $(i,i+1)$ we include all Pauli operators supported only on those $2\ell$ qubits that comprise the pair. The recovery algorithm reconstructs the snapshots and aggregates estimations via $\cfont{MoM}$ technique.\\
  \\
  For notational convenience, define
\[
\rho_{i,i+1}:=\Tr_{\overline{i,i+1}}(\rho), \quad \widetilde{\sigma}_{i,i+1}:= \Tr_{\overline{i,i+1}}\left(\frac1L\sum_{l\in[L]}s_l\right).
\]
\textbf{Completeness.}
If the $\iDCLDM$ instance is a YES instance, there exists a state $\rho$ such that
\[
\|\rho_{i,i+1}-\sigma_{i,i+1}\|_{\Tr}\le \alpha
\qquad \forall i\in[n-1].
\]
By construction of the integer counts and the stitched shadow, we have
\[
\norm*{\widetilde{\sigma}_{i,i+1}-\sigma_{i,i+1}}_{\Tr}
\le \epsilon
\qquad \forall i\in[n-1].
\]
Therefore, by the triangle inequality, $\norm*{\rho_{i,i+1}-\widetilde{\sigma}_{i,i+1}}_{\Tr}\le \alpha+\epsilon$. For every Pauli observable $P$ on the qubits of the neighboring pair $(i,i+1)$, we have $\|P\|_\infty=1$, and hence by Hölder's inequality,
\[
\begin{aligned}
\left|\Tr(P\rho)-A(S,P)\right|
&=
\left|
\Tr\left(P\left(\rho_{i,i+1}-\widetilde{\sigma}_{i,i+1}\right)\right)
\right|\\
&\le
\norm*{\rho_{i,i+1}-\widetilde{\sigma}_{i,i+1}}_{\Tr}
\le \alpha+\epsilon .
\end{aligned}
\]
Thus the constructed $\CSV_{\cfont{HKP}}$ instance is a YES instance with $\alpha_{\mathrm{CSV}}:=\alpha+\epsilon$.\\
\textbf{Soundness.}
For soundness, suppose for contradiction that there exists a state $\rho$ such that, for every Pauli observable $P$ supported on a neighboring pair, $\left|\Tr(P\rho)-A(S,P)\right|\le \beta_{\mathrm{CSV}}$. Equivalently,
\[
\left|
\Tr\left(P\left(\rho_{i,i+1}-\widetilde{\sigma}_{i,i+1}\right)\right)
\right|
\le \beta_{\mathrm{CSV}}
\qquad \forall i\in[n-1].
\]
Since the Pauli observables on the $2\ell$ qubits of a neighboring pair form a tomographically complete operator basis, and since the local dimension $d^2=2^{2\ell}$ is constant, there exists a constant $c_d>0$ such that
\[
\norm*{\rho_{i,i+1}-\widetilde{\sigma}_{i,i+1}}_{\Tr}
\le c_d\beta_{\mathrm{CSV}}
\qquad \forall i\in[n-1].
\]
By construction of the integer counts and the stitched shadow,
\[
\norm*{\widetilde{\sigma}_{i,i+1}-\sigma_{i,i+1}}_{\Tr}
\le \epsilon
\qquad \forall i\in[n-1].
\]
Therefore, by the triangle inequality,
\[
\norm*{\rho_{i,i+1}-\sigma_{i,i+1}}_{\Tr}
\le c_d\beta_{\mathrm{CSV}}+\epsilon
\qquad \forall i\in[n-1].
\]
Choose $\beta_{\mathrm{CSV}}$ so that $c_d\beta_{\mathrm{CSV}}+\epsilon < \beta$, where $\beta$ is the NO threshold of the input $\iDCLDM$ instance. Then $\rho$ would be a state whose every neighboring marginal is within distance strictly less than $\beta$ of the corresponding $\sigma_{i,i+1}$, contradicting the NO case of the $\iDCLDM$ instance. Hence the constructed $\CSV_{\cfont{HKP}}$ instance is a NO instance.
\end{proof}

\begin{corollary}\label{cor:CSVHKP}
  $\QMA\le \CSV_{\text{HKP}}$, even for $6$-local observables on a spatially sparse hypergraph.
\end{corollary}
\begin{proof}
  Follows from \cref{Thm:iDCLDM=QMA,thm:1DCLDM-CSV}. Since here $d=8$ and so $\ell=3$, the observables are $6$-local on qubits. It is easy to verify that the resulting hypergraph (where each qubit is a vertex, and each Pauli operator acting non-trivially on a set $V'\subseteq V$ of vertices is represented by a hyperedge) is spatially sparse, as per \Cref{def:sparse}.
\end{proof}
\begin{corollary}\label{cor:CSVHKPisQMAcomplete}
  $\CSV_{\text{HKP}}$ is $\QMA$-complete, even for $6$-local observables on a spatially sparse hypergraph.
\end{corollary}
\begin{proof}
    Hardness follows from \cref{cor:CSVHKP} above and containment from  \cref{rem:CSVp-to-ObsConp} and the facts that $\ObsConPoly\in\SuperQMAPoly$ (\cref{prop:obsconpoly-in-superqmapoly}) and $\SuperQMAPoly=\QMA$ (\cref{Thm:QMAplus=QMA}).
\end{proof}

\subsection{MYZ classical shadow}
Our hardness result is not limited to qubit translated systems. A recent protocol by Mao, Yi, and Zhu \cite{MYZ25} generalizes the local Clifford measurement framework to qudits of odd prime dimension $d$. Their protocol uses the ensemble $\mathcal{E} = \mathrm{Cl}(d)^{\otimes n}$, where $\mathrm{Cl}(d)$ is the single-qudit Clifford group, leading to snapshots that are tensor products of single-qudit operators. Each such operator is derived from one of the $d(d+1)$ single-qudit stabilizer states (for details see \cref{par:MYZ}). 
\begin{definition}[$\CSV_{\cfont{MYZ}}(d)$]\label{def:csvmyz}
The definition is the same as \cref{def:CSV}, only now the classical shadow has the structure dictated by the MYZ local-Clifford protocol on odd-prime $d$. Concretely:
\begin{itemize}
  \item Shadow $S$ consists of $L=\poly(n)$ strings. Each string is of the form
  $((\mu,b)_1,\ldots,(\mu,b)_n)$ with $\mu\in \mathbb{F}_d \cup \{\infty\}$ the measurement basis label and $b\in[d]$ the measurement outcome. 
  \item $O$ is a set of $k$-local observables on $n$ qudits, for fixed $k=O(1)$.
  \item The recovery algorithm applies the inverse channel of the measurement protocol on $S$ to get the snapshots and then aggregates via MoM, subject to the range convention of \Cref{rem:range-convention}.
\end{itemize}
\end{definition}
This protocol works for odd prime $d$. Notice that we can always pad the local dimensions of our chain and add projector terms in our Hamiltonian and so we can trivially get a $\QMA$-completeness under Karp reductions result for a $d\ge8$-level $\iDCLDM$ problem.
\begin{theorem}\label{thm:1DCLDM-CSVMYZ}
  For every fixed odd prime local dimension $d\ge 11$,$\;\;\iDCLDM_{d}\le\CSV_{\text{MYZ}}(d)$.
\end{theorem}
\begin{proof}[Proof sketch]
The proof is analogous with \cref{thm:1DCLDM-CSV}, only here the local dimension of the qudits is an odd prime. Let us first quickly summarize the differences of the two: 
\begin{itemize}
    \item \textbf{Single-site snapshot types ($\hat{\eta}$):} In MYZ local-Clifford ensemble, each site is measured in one of the $d+1$ stabilizer basis- the eigenbases of $Z$ and $Z^tX$ for $t\in\F_d$. For basis label $\mu\in \mathbb{F}_d \cup \{\infty\}$ and outcome label $b\in\F_d$, the single-site snapshot operator is
    \[
    \hat{\eta}_{\mu,b}=(d+1)\ket{\phi_{\mu,b}}\bra{\phi_{\mu,b}}-I_d
    \]
    where $\ket{\phi_{\mu,b}}$ is the eigenvector in the basis labelled by $\mu$ with outcome label $b$. The alphabet size is now $m'=d(d+1)$, so still constant for fixed $d$.
    \item \textbf{Observables:} Our observables will now be all the Hermitian real and imaginary parts of generalized Pauli/Weyl operators supported on adjacent qudits. 
\end{itemize}
With these changes in mind we can see that our proof follows directly. Since the alphabet $m'$ is still constant we can solve \cref{eq:marginal-system2} system efficiently, via the same DP algorithm. After that, we use the same ``stitching the local shadows" argument to create a global shadow which alongside our observables and the known recovery algorithm will form the $\CSV_{\cfont{MYZ}}$ instance.
\end{proof}
\begin{corollary}\label{cor:CSVMYZ}
$\QMA\le\CSV_{\cfont{MYZ}}(d)$ for every fixed odd prime local dimension $d\geq 11$, even for $2$-local nearest-neighbor observables on a line.
\end{corollary}
\begin{proof}
 Follows from  \cref{thm:1DCLDM-CSVMYZ,Thm:iDCLDM=QMA}.
\end{proof}
\begin{corollary}\label{cor:CSVMYZisQMAcomplete}
  For every fixed odd prime local dimension $d\ge 11$,
  $\CSV_{\mathrm{MYZ}}(d)$ is $\QMA$-complete, even for $2$-local
  nearest-neighbor observables on a line.
\end{corollary}

\begin{proof}
Hardness follows from \cref{cor:CSVMYZ}. For containment, observe that
$\CSV_{\mathrm{MYZ}}(d)$ has polynomially many observables and is a
fixed-local-dimension instance of $\CSV$. Since $d$ is fixed, each qudit
can be encoded into $\lceil \log d\rceil=O(1)$ qubits, and the MYZ
observables and recovery procedure remain efficiently implementable.
Thus $\CSV_{\mathrm{MYZ}}(d)$ reduces to $\ObsConPoly$ by
\cref{rem:CSVp-to-ObsConp}. Finally,
$\ObsConPoly\in\SuperQMAPoly$ by
\cref{prop:obsconpoly-in-superqmapoly}, and
$\SuperQMAPoly=\QMA$ by \cref{Thm:QMAplus=QMA}. Hence
$\CSV_{\mathrm{MYZ}}(d)\in\QMA$.
\end{proof}

\subsection{``Dequantizing'' HKP, MYZ for global Clifford measurements}

Recall now that classical shadows constructed using global Clifford operations allow for an efficient recovery of the expectation values of observables whose Frobenius norm $\lVert O \rVert_{F} = \sqrt{\Tr[ O^{\dagger} O]}$ is bounded. Interestingly, in this setting, we can solve the validity problem in polynomial time if we have sampling and query access to the target observables. Briefly, this is done by invoking the bounded Frobenius norm semidefinite programming (SDP) dequantization result of \cite{CGLLTW22} (see also \cite{CLLW20}, which previously handled the low-rank case).

We begin by defining the Global Clifford version of CSV. 
\begin{definition}[$\CSV_{\cfont{GC}}$]\label{def:CSV_GC}
     The definition is the same as \cref{def:CSV}, only now our classical shadow has the structure dictated by the global Clifford measurement protocol presented in Ref.\ \cite{HKP20}. That means the following:
    \begin{itemize}
        \item The shadow $S$ consists of $L=\poly(n)$ strings, $\{s_i\}_{i=1}^L$, each of which encodes the random $n$-qubit Clifford used in that round and the measurement outcome. More formally, each string will be of the form $s_i=(\cfont{Stab}_i,b_i)$ where $\cfont{Stab}$ is the efficient classical representation of the global Clifford via the stabilizer formalism and $b\in \{0,1\}^n$ the measurement outcome of that round.
        \item $O$ is any set of $\poly(n)$ observables with bounded Frobenius norm, i.e., $\|O_i\|_{F}\le \poly(n)$. Those observables are possibly highly non-local.
        \item The recovery algorithm applies the global inverse depolarizing channel to extract the snapshot operators $\hat{\eta}$ from $S$ and aggregates the estimates via the \emph{median of means} (MoM) technique. The output is subject to the range convention of \Cref{rem:range-convention}.
    \end{itemize}
\end{definition}
It will be easier to work in the abstract definition which we denote $\ObsCon_{\cfont{F}}$ and define as:
 \begin{definition}[$\ObsCon_{\cfont{F}}$]\label{def:ObsCon_HS}
  The input is a set of observables along with their respective expectation values $(O_i, y_i)_{i=1}^{m=\poly(n)}$, with $\|O_i\|_F\le \poly(n)$, and parameters $\alpha$ and $\beta$ satisfying $\beta-\alpha\geq 1/\poly(n)$. The output is to decide between the following cases:
\begin{itemize}
    \item \textbf{Yes}: $\exists$  $n$-qubit state $\rho$ s.t.\ $\forall$ $i\in [m]$, 
    $\left| \Tr\left(O_i\rho\right)-y_i\right|\leq \alpha$.
    \item \textbf{No}: $\forall$ $n$-qubit states $\rho$ $\exists$ some $i\in [m]$ s.t.\  $\left| \Tr \left(O_i\rho\right)-y_i\right|\geq \beta$.
\end{itemize} 
We assume $y_i\in[-1,1]$.
 \end{definition}
Let us now properly define what a sampling and query access to the target observables mean.

\begin{definition}[Sampling and query access~\cite{CGLLTW22}]\label{def:SQAccess}
For a vector $v\in\C^N$, sampling and query access, denoted $\SQ{v}$, means that we can query entries $v(i)$, sample indices $i\in[N]$ with probability $\frac{\abs{v(i)}^2}{\norm{v}_2^2}$, and compute $\norm{v}_2$. Query access alone, denoted $Q(v)$, means that we can query entries $v(i)$. For a matrix $A\in\C^{M\times N}$, query access, denoted $Q(A)$, means that given $(i,j)\in[M]\times[N]$ one can compute $A(i,j)$; sampling and query access, denoted $SQ(A)$, means that we have $SQ$-access to each row of $A$ and $SQ$-access to the vector of row norms of $A$.
\end{definition}
 With our sampling and query access definitions in hand, we can define:
  \begin{definition}[$\ObsCon_{\cfont{F,Samp}}$]\label{def:ObsCon_Fsamp}
   Defined as $\ObsCon_{\cfont{F}}$ (see \cref{def:ObsCon_HS}) but additionally with sampling and query access (see \cref{def:SQAccess}) to each observable $O_i$.
  \end{definition}

We now recall the definition of the (SDP $\varepsilon$-feasibility)-problem~\cite{CLLW20,CGLLTW22}.

\begin{definition}[SDP $\varepsilon$-feasibility \cite{CGLLTW22}]\label{def:SDP} Given an $\varepsilon>0$, $m$ real numbers $b_1,\dots,b_m\in \mathbb R$, and Hermitian $N\times N$ matrices $SQ(A^{(1)}),\dots,SQ(A^{(m)})$ such that $-I\preceq A^{(i)}\preceq I$ for all $i\in[m]$, we define $\mathcal{S}_{\varepsilon}$ as the set of all $\rho$ satisfying
\begin{subequations}
  \begin{align}
    \Tr[A^{(i)}\rho]&\le b_i+\varepsilon,\;\; \forall i\in[m]\\
    \rho&\succeq 0\\
    \Tr[\rho]&=1
  \end{align}
  \end{subequations}
If $\mathcal{S}_{\varepsilon}=\emptyset$, output ``infeasible". If $\mathcal{S}_0\neq\emptyset$, output a $\rho\in\mathcal{S}_\varepsilon$.
    
\end{definition}
\begin{lemma}\label{lem:SDPreduct}
    $\mathsf{ObsCon_{F,Samp}}\le \text{SDP} \;\varepsilon\text{-feasibility}$.
\end{lemma}
\begin{proof}
    We start with a $\mathsf{ObsCon_{F,Samp}}$ instance: $\{\left(\SQ(O_i),y_i\right)\}_{i=1}^m$ with $\|O_i\|_\infty\le 1,\;\|O_i\|_F\le\poly(n)$ and parameters $\alpha,\;\beta$ with $\beta-\alpha\ge1/\poly$. Now define:
    \[
    A_{i,\pm}:=\pm O_i,\;\;b_{i,\pm}:=\alpha\pm y_i,\;\;
    \varepsilon:=\frac{\beta-\alpha}{2}
    \]
    Consider now the following SDP
    \begin{subequations}
  \begin{align}
    \Tr[A_{i,\pm}\rho]&\le b_{i,\pm}+\varepsilon,\;\; \forall i\in[m]\\
    \rho&\succeq 0\\
    \Tr[\rho]&=1
  \end{align}
  \end{subequations}
Since $\|O_i\|_\infty\le 1$, one has $-I\preceq A_{i,\pm}\preceq I$ for all $i$. Thus, the above SDP is a valid instance of the $\text{SDP} \;\varepsilon\text{-feasibility}$ problem. We now show correctness.

\textbf{Completeness:} Assume the $\mathsf{ObsCon_{F,Samp}}$ instance is a YES instance, so that there exists a state $\rho$ s.t. 
$ |\Tr(O_i\rho)-y_i|\le\alpha\;\;\forall\;i\in[m]$.
Equivalently:
\[
\exists\rho\;\;\text{s.t.}\;\Tr(O_i\rho)\le y_i+\alpha,\;\; \Tr(-O_i\rho)\le\alpha-y_i\;\;\forall\; i\in[m].
\]
In terms of the SDP constraints: $\exists\;\rho\;\text{s.t.}\;\;\Tr(A_{i,\pm}\rho)\le b_{i,\pm}\;\;\forall\;i\in[m]$. Therefore the associated SDP instance is feasible with zero slack, i.e., $\mathcal{S}_0\neq\emptyset$.

\textbf{Soundness:} Assume $\mathsf{ObsCon_{F,Samp}}$ instance is a NO instance. Then for every state $\rho$ there exists some index $i^*\in[m]$ such that $|\Tr(O_{i^*}\rho)-y_{i^*}|\ge\beta$. Let $g:=\beta-\alpha$. Then for the index $i^*$, one of the following must hold 
\[
\Tr(O_{i^*}\rho)-y_{i^*}\ge\alpha+g\;\;\text{or}\;\; -(\Tr(O_{i^*}\rho)-y_{i^*})\ge\alpha+g
\]
Rewriting these in terms of the SDP constraints we get: $\Tr(A_{i^*,\pm}\rho)\ge b_{i^*,\pm}+g$. So every $\rho$ violates at least one SDP constraint by at least $g=\beta-\alpha=2\varepsilon$. Therefore $\mathcal{S}_{\varepsilon}=\emptyset$. 
\end{proof}
\begin{lemma}[Corollary 6.25 \cite{CGLLTW22}]\label{cor:SDPsol} Let $F\ge\max_{j\in[m]}(\|A_j\|_F)$, and suppose $F=\Omega (1)$. Then we can solve \ref{def:SDP} with success probability $\ge 1-\delta$ in cost
\[
\widetilde{\mathcal{O}}\left(\left(\frac{F^{18}}{\varepsilon^{40}}\log^{20}(N)\textbf{sq}(A)+\frac{F^{22}}{\varepsilon^{46}}\log^{23}(N)+m\frac{F^8}{\varepsilon^{18}}\log^{8}(N)\textbf{q}(A)+m\frac{F^{14}}{\varepsilon^{28}}\log^{13}(N)\right)\log^3\frac{1}{\delta}\right)
\]
providing sampling and query access to a solution.   
\end{lemma}
\begin{theorem}\label{thm:csvgc}
$\ObsCon_{\cfont{F,Samp}}$, and hence $\CSV_{\cfont{GC}}$ under query and sampling access,
is solvable in randomized classical polynomial time.
\end{theorem}
\begin{proof}
    From \cref{lem:SDPreduct,cor:SDPsol} and \cref{rem:CSVp-to-ObsConp}, since $\|O_i\|_F\le\poly(n)\;\forall\;i\in[m]$. 
\end{proof}

We now state the qudit analogue for the global $n$-qudit Clifford ensemble protocol (see \cref{par:MYZ}).
\begin{definition}
 [$\ObsCon_{\cfont{F,Samp}}^{(d)}$]$\ObsCon_{\cfont{F,Samp}}^{(d)}$ is the qudit analogue of $\ObsCon_{\cfont{F,Samp}}$ (see \cref{def:ObsCon_Fsamp}): inputs $(O_i,y_i)_{i=1}^m$ with Hermitian $n$-qudit $O_i$ satisfying $\|O_i\|_{F}\le\poly(n)$ and $\beta-\alpha\ge1/\poly(n)$, together with sampling and query access.  
\end{definition}
\begin{theorem}\label{thm:csvgcmyz}
 $\ObsCon_{\cfont{F,Samp}}^{(d)}$ is solvable in randomized classical polynomial time.   
\end{theorem}
\begin{proof}
    The proof of \Cref{thm:csvgc} is not specific to qubits. If $N$ denotes the Hilbert-space dimension of the SDP variable, then the SDP solver depends only polylogarithmically on $N$. Replacing the qubit dimension $N=2^n$ by the qudit dimension $N=d^n$ therefore changes the logarithmic dimension factor from $\log(2^n)=n$ to $\log(d^n)=n\log d$. Hence, for fixed local dimension $d$ and $\|O_i\|_F\le \poly(n)$, the same reduction to SDP $\varepsilon$-feasibility gives a randomized classical polynomial-time algorithm for $\ObsCon_{\cfont{F,Samp}}^{(d)}$.
\end{proof}

\section{Product state variants and connections to $\QMAt$}\label{scn:product}

In this section we explore the variants of our problems stemming from the restriction to the product-state space and show completeness results for the corresponding product-state classes, i.e., $\PSQMA$, $\Pqc$ (See \cref{def:PSQMA,def:P-qc-sigma2} respectively) in the $\poly$ and the $\exp$ regime.

Let us start by analyzing the $\poly$ case. First let us show the equivalence between $\PSQMA_{\poly}$ and $\QMAt$:

\begin{lemma}\label{lem:QMAii-subset-PSQMA}
 $\QMAt\subseteq \PSQMA_{\poly}$.  
\end{lemma}
\begin{proof}
The proof is completely analogous to Lemma 4.2 in Ref.\ \cite{aharonovLatticeProblemQuantum2003}.
Given a verifier $V$ for $L\in\QMAt$ construct a super-verifier that outputs $(V,r=1,s=\frac{1}{3})$. This, as we will see, satisfies the definition of $\PSQMA_{\poly}$ by using $\epsilon=\frac{1}{6}$, $m=1$.\\
\textbf{Completeness:} Let $ x\in L$ then  $\exists\rho_A\otimes\rho_B$  s.t. $ \Tr(\Pi^{(1)}V(\rho_A\otimes\rho_B) V^{\dagger})\geq2/3$. It is easy to see that the condition $|\Tr(\Pi^{(1)}V(\rho_A\otimes\rho_B) V^{\dagger})-1|\leq1/3$ holds. \\
\textbf{Soundness:} Let $x\notin L$ then $\forall\rho_A\otimes\rho_B,\; \Tr(\Pi^{(1)}V(\rho_A\otimes\rho_B) V^{\dagger})\leq1/3$. This means that the condition  $|\Tr(\Pi^{(1)}V(\rho_A\otimes\rho_B) V^{\dagger})-1|\leq1/3+1/6=1/2$ is never satisfied. Notice that here we only have one check ($m=1$) and so that check must fail in the No case, as it does.
\end{proof}

\begin{lemma}\label{lem:PSQMA-subset-QMAii}
    $ \PSQMA_{\poly}\subseteq \QMAt$.
\end{lemma}
\begin{proof}
Let us use another characterization of $\QMAt$ called $\SymQMA$. This class was defined in Ref.\ \cite{ABDFS09} where we have the promise that the $k$-unentangled proofs are all the same. Aaronson et al. proved that $\QMAt=\SymQMA$ under the $\QMAt$ amplification conjecture that was resolved in the positive in Ref.\ \cite{HM10}.\\
In order to simulate the $\PSQMA_{\poly}$ protocol with a $\SymQMA$ protocol, we do the following:
\begin{itemize}
    \item with $1/2$ probability we pick a random pair of witnesses and run the product test (Protocol 1 \cite{HM10}). Accept iff the product test outputs ``product".
    \item with $1/2$ probability pick $i\in[m]$ uniformly at random and run $V_i$ on all k-copies. Let $r'$ be the number of $1$'s measured divided by $k$. Accept iff $|r'-r_{x,i}|\le s_{x,i}+\frac{\epsilon}{2}$.
\end{itemize}
\textbf{Completeness:} We know from the promise of $\PSQMA_{\poly}$ there exists $(\rho_A\otimes\rho_B)$ for which $|\Tr(\Pi^{(1)}V_{x,i}(\rho_A\otimes\rho_B)V_{x,i}^{\dagger})-r_{x,i}|\le s_{x,i},\;\forall i\in[m]$. The verifier expects purifications $\ket{\Psi}_{AR_A}\otimes \ket{\Phi}_{BR_B}$ (i.e., $\Tr_{R_A}\ket{\Psi}\bra{\Psi}=\rho_A,\; \Tr_{R_B}\ket{\Phi}\bra{\Phi}=\rho_B$). In this case, step 1 (product test across $(AR_A):(BR_B)$) accepts with probability 1. In step 2 the verifier traces out $R_A,R_B$ and runs $V_{x,i}$ only on $A,B$. This preserves exactly the target statistic on $\rho_A\otimes\rho_B$. Now according to the Hoeffding bound, for $k=\Theta(n/\epsilon^2)$, the probability that $|r'-\Tr(\Pi^{(1)}V_{x,i}(\rho_A\otimes\rho_B)V_{x,i}^{\dagger})|\le \epsilon/2$ is at least $1-2^{-\Omega(n)}$. Thus, we have
\[
|r'-r_{x,i}|\le s_{x,i}+\frac{\epsilon}{2}
\]
with probability at least $1-2^{-\Omega(n)}$. That leads to a total acceptance probability of
\[
p_{\cfont{acc}}\ge\frac{1}{2}+\frac{1}{2}(1-2^{-\Omega(n)}).
\]
\textbf{Soundness:} 
Let $\ket{\Xi}_{AR_ABR_B}$ be an arbitrary pure state sent as the single-copy symmetric witness, and let 
\[
\gamma:=1-\max_{\ket{\alpha},\ket{\beta}}|\bra{\Xi}\ket{\alpha}_{AR_A}\otimes\ket{\beta}_{BR_B}|^2.
\]
be its infidelity to the closest pure product state across the cut $(AR_A):(BR_B)$.
\begin{itemize}
\item By Theorem 1 in Ref.\ \cite{HM10}, the product test rejects with probability at least $\frac{11}{512}\gamma$.
\item By Lemma 22 in Ref.\ \cite{HM10} we have that for any $0\le P\le I$:
\[
|\bra{\Xi}P\ket{\Xi}-\bra{\alpha\otimes\beta}P\ket{\alpha\otimes\beta}|\le \sqrt\gamma
\]
In particular, for $P=(V_{x,i}^{\dagger}\Pi^{(1)}V_{x,i})_{AB}\otimes I_{R_AR_B}$ the single-copy acceptance probability of check $i$ on $\ket{\Xi}$ differs from that on the closest product state by at most $\sqrt{\gamma}$.
\end{itemize}
We distinguish two cases.

\textbf{Case 1}: $\ket{\Xi}$ is far from a product state. Suppose $\sqrt{\gamma}>\epsilon/4$ and so $\gamma>\epsilon^2/16$. Therefore the product test branch rejects with probability $\frac{11}{512} \gamma>\Omega(\epsilon^2).$

\textbf{Case 2}: $\ket{\Xi}$ is close to a product state. Suppose $\sqrt{\gamma}\le \frac{\epsilon}{4}$. Let $\ket{\alpha}_{AR_A}\otimes\ket{\beta}_{BR_B}$ be the closest product state to $\ket{\Xi}$. Since we are in the NO case of $\PSQMA_{\poly}$, at least $1/m$ fraction of indices $i\in[m]$ are bad for this product state by $\epsilon$. That is, for those bad $i$,
\[
|\Tr(\Pi^{(1)}V_{x,i}(\rho_A\otimes\rho_B)V_{x,i}^{\dagger})-r_{x,i}|\ge s_{x,i}+\epsilon
\]

Lemma 22 in \cite{HM10} gives $\left|
\bra{\Xi}P_i\ket{\Xi}-\Tr(\Pi^{(1)}V_{x,i}(\rho_A\otimes\rho_B)V_{x,i}^{\dagger})
\right|
\le \frac{\epsilon}{4}$.

Hence, by the triangle inequality,
\[
\begin{aligned}
|\bra{\Xi}P_i\ket{\Xi}-r_{x,i}|
&\ge
\left|\Tr(\Pi^{(1)}V_{x,i}(\rho_A\otimes\rho_B)V_{x,i}^{\dagger})-r_{x,i}\right| -
\left|
\bra{\Xi}P_i\ket{\Xi} - \Tr(\Pi^{(1)}V_{x,i}(\rho_A\otimes\rho_B)V_{x,i}^{\dagger})
\right|  \\
&\ge s_{x,i}+\epsilon-\frac{\epsilon}{4}
=
s_{x,i}+\frac{3\epsilon}{4}.
\end{aligned}
\]

Now for the verifier in the second branch to accept, conditioned on choosing a bad index $i$, the empirical average must deviate from its true mean by at least $\epsilon/4$. By Hoeffding's inequality this happens with probability at most $2^{-\Omega(k\epsilon^2)}$. Since a uniformly random index $i\in [m]$ is bad with probability at least $1/m$, the second branch rejects with probability at least $\frac{1}{m}(1-2^{-\Omega(k\epsilon^2)})$.

Thus, a NO instance is rejected with inverse polynomial probability while a YES instance is accepted with probability exponentially close to 1. Applying the Harrow-Montanaro amplification theorem of $\QMAt$ completes the proof.
\end{proof}   
\begin{corollary}\label{cor:PSQMA=QMAii}
    $\PSQMA_{\poly}=\QMAt$.
\end{corollary}
\begin{proof}
 Follows from  \cref{lem:QMAii-subset-PSQMA,lem:PSQMA-subset-QMAii}.  
\end{proof}
\begin{definition}
    [$\ProductObsCon_{\poly}$] Define $\ProductObsCon_{\poly}$ as in \cref{def:obscon} with $\rho=\rho_A\otimes\rho_B$ and $m=\poly(n)$.
\end{definition}
\begin{lemma}\label{lem:PObsCon-complete-PSQMA}
    $\ProductObsCon_{\poly}$ is $ \PSQMA_{\poly}$-complete.
\end{lemma}
\begin{proof}
    The containment is shown as in \cref{prop:obsconpoly-in-superqmapoly} and the hardness as in
    \cref{prop:obsconpoly-hard-superqmapoly}. Replace $\rho$ with $\rho_A\otimes\rho_B$ and everything else follows as is.
\end{proof}
\begin{corollary}
$\ProductObsCon_{\poly}$ is $\QMAt$-complete.
\end{corollary}
\begin{proof}
Follows from \cref{lem:PObsCon-complete-PSQMA,cor:PSQMA=QMAii}.
\end{proof}
Now in order to see what happens in the case of exponentially many checks we follow the same pattern as before and make use of the classes $\PSQMA_{\exp}$ and $\Pqc$. 

We start by showing the equivalence of the two classes.

\begin{lemma}\label{lem:Pqc-in-PUSQMA}
    $\Pqc\subseteq\PSQMA_{\exp}$.
\end{lemma}
\begin{proof}
    Given a verifier $V$ for a language $L\in \Pqc(c,s)$, where $\Delta:=c-s\ge 1/\poly(n)$, we construct a product-state super-verifier $V'$ for $L$. Start by hardwiring the classical proof $y$ into the verifier $V$ of $\Pqc$, let us call it $V_y$. The super-verifier's checks are now parametrized by the classical strings $y\in\{0,1\}^{q(n)}$, i.e., $m=2^{q(n)}$. Construct a super-verifier $V'$ that on input $x$ picks uniformly at random a challenge $y$ and outputs the check $\left(V_y, r=1, s'=1-c\right)$. This satisfies the definition of $\PSQMA_{\exp}$ with $\epsilon=(c-s)/2$.\\
    \textbf{Completeness:} Let $ x\in L$ then  $\exists\rho_A\otimes\rho_B$ such that $\forall y\; \Tr(\Pi^{(1)}V_y(\rho_A\otimes\rho_B) V_y^\dagger)\geq c$. It is easy to see that the condition $\left|\Tr(\Pi^{(1)} V_y(\rho_A\otimes\rho_B)V_y^{\dagger})-1\right|\leq1-c=s'$
    is satisfied for all $y$.\\
\textbf{Soundness:} Let $x\notin L$ then $\forall\rho_A\otimes\rho_B\;\exists y$ s.t. $\Tr(\Pi^{(1)}V_y(\rho_A\otimes\rho_B)V_y^\dagger)\leq s$.
This means that the condition $\left|\Tr(\Pi^{(1)}V_y(\rho_A\otimes\rho_B)V_y^{\dagger})-1\right|
\geq 1-s$
is satisfied for at least one $y$, for each $\rho_A\otimes\rho_B$. Since $1-s=(1-c)+(c-s)> (1-c)+\frac{c-s}{2}=s'+\epsilon$,
we get that the condition is violated by more than $s'+\epsilon$ for at least one $y$, for each $\rho_A\otimes\rho_B$, so with probability $\geq \frac{1}{m}$.
\end{proof}
\begin{lemma}\label{lem:PUSQMA-in-Pqc}
    $\PSQMA_{\exp}\subseteq\Pqc$.
\end{lemma}
\begin{proof}
Given a product-state super-verifier for $L\in\PSQMA_{\exp}$, say $\{(V_{x,i},r_{x,i},s_{x,i})\}_{i\in[m]}$ where $m=\exp(n)$ we construct a $\Pqc$ verifier for $L$. For a product state $\rho_A\otimes\rho_B$ define the acceptance probability of check $i$ as:
\[
p_i(\rho_A,\rho_B):=\Tr\left(\Pi^{(1)}V_{x,i}(\rho_A\otimes\rho_B)V_{x,i}^{\dagger}\right)
\]
The existential proof is a product state $\rho_A\otimes\rho_B$ and the universal proof is a pair $(i,b)$, where $i\in[m]$ and $b\in\{+,-\}$. The sign $b$ tells the verifier whether the violation is above or below the allowed interval. For each $i$ define:
\[
u_i:=\min\{1,r_{x,i}+s_{x,i}\},\quad \ell_i:=\max\{0,r_{x,i}-s_{x,i}\}
\]
So the YES interval is in $[\ell_i,u_i]$.

If $b=+$, the verifier wants to punish the case where $p_i(\rho_A,\rho_B)>u_i$. For that it runs $V_{x,i}$ on $\rho_A\otimes\rho_B$, obtains an output bit $X\in\{0,1\}$ and then postprocesses as follows:
\[
\Pr[\text{accept}|X=0]=\frac{1+u_i}{2},\quad \Pr[\text{accept}|X=1]=\frac{u_i}{2}
\]
So the total acceptance probability is $q_i^+(\rho_A,\rho_B)=\frac{1+u_i-p_i(\rho_A,\rho_B)}{2}$.

If $b=-$, the verifier wants to punish the case where $p_i(\rho_A,\rho_B)<\ell_i$. For that it runs $V_{x,i}$, obtains $X\in\{0,1\}$, and postprocesses as follows:
\[
\Pr[\text{accept}|X=0]=\frac{1-\ell_i}{2},\quad \Pr[\text{accept}|X=1]=1-\frac{\ell_i}{2}
\]
Therefore the total acceptance probability is $q_i^-(\rho_A,\rho_B)=\frac{1-\ell_i+p_i(\rho_A,\rho_B)}{2}$.

\textbf{Completeness}: Suppose $x\in L$. Then, by the $\PSQMA_{\exp}$ completeness condition, there exists a product state $\rho_A\otimes\rho_B$ such that for every $i\in[m]$, $|p_i(\rho_A,\rho_B)-r_{x,i}|\le s_{x,i}$ and so $\ell_i\le p_i(\rho_A,\rho_B)\le u_i$ for every $i$. Now consider any universal proof $(i,+)$. Since $p_i(\rho_A,\rho_B)\le u_i$ we have
\[
q_i^+(\rho_A,\rho_B)=\frac{1+u_i-p_i(\rho_A,\rho_B)}{2}\ge\frac{1}{2}.
\]
Similarly, for any universal proof $(i,-)$, since $p_i(\rho_A,\rho_B)\ge\ell_i$ we have
\[
q_i^-(\rho_A,\rho_B)=\frac{1-\ell_i+p_i(\rho_A,\rho_B)}{2}\ge\frac{1}{2}.
\]
Thus there exists a product witness such that for every universal classical proof, the verifier accepts with probability at least $c'=1/2$.

\textbf{Soundness}: Suppose $x\notin L$. Then, by the $\PSQMA_{\exp}$ soundness condition, for every product state $\rho_A\otimes\rho_B$, there exists some check $i\in[m]$ such that $|p_i(\rho_A,\rho_B)-r_{x,i}|\ge s_{x,i}+\epsilon$. There are two cases:

\textbf{Case 1}: The acceptance probability is too high, i.e.,  $p_i(\rho_A,\rho_B)\ge r_{x,i}+s_{x,i}+\epsilon$ and so $u_i=r_{x,i}+s_{x,i}$. The universal prover sends $(i,+)$. Then 
\[
q_i^+(\rho_A,\rho_B)=\frac{1+u_i-p_i(\rho_A,\rho_B)}{2}\le\frac{1}{2}-\frac{\epsilon}{2}
\]
\textbf{Case 2}: The acceptance probability is too low, i.e., $p_i(\rho_A,\rho_B)\le r_{x,i}-s_{x,i}-\epsilon$ and so $\ell_i=r_{x,i}-s_{x,i}$. The universal prover sends $(i,-)$. Then 
\[
q_i^-(\rho_A,\rho_B)=\frac{1-\ell_i+p_i(\rho_A,\rho_B)}{2}\le\frac{1}{2}-\frac{\epsilon}{2}.
\]
Thus, for every product state $\rho_A\otimes\rho_B$, there exists a universal classical proof $(i,b)$ such that the verifier accepts with probability at most $s'=1/2-\epsilon/2$. Therefore the gap is
\[
c'-s'=\frac{1}{2}-(\frac{1}{2}-\frac{\epsilon}{2})=\frac{\epsilon}{2}\ge \frac{1}{\poly(n)}
\] 
Hence $L\in\Pqc$.
\end{proof}
\begin{corollary}\label{cor:pqc-equals-psqmaexp}
  $\Pqc=\PSQMA_{\exp}$. 
\end{corollary}
\begin{proof}
    Follows from \cref{lem:Pqc-in-PUSQMA,lem:PUSQMA-in-Pqc}.
\end{proof}
\begin{definition}
    [$\ProductObsCon_{\exp}$] Define $\ProductObsCon_{\exp}$ as 
    in \cref{def:obscon} with $\rho=\rho_A\otimes\rho_B$ and $m=\exp(n)$.
\end{definition}
\begin{lemma}\label{lem:PObsCon-eq-PUSQMA}
 $\ProductObsCon_{\exp}$ is $ \PSQMA_{\exp}=\Pqc$-complete.  
\end{lemma}
\begin{proof}
  The containment proof is similar as in \cref{prop:obsconpoly-in-superqmapoly} and the hardness proof as in \cref{prop:obsconexp-hard-superqmaexp}. The only difference is that instead of a state $\rho$, we have a product state $\rho_A\otimes\rho_B$. 
\end{proof}
\section{Variants of $\CSV$: Robustness and multiple shadow consistency}\label{scn:variant}

In this section we introduce two natural variants of $\CSV$: a \emph{randomized} (sampled-shadow) formulation and a \emph{multiple-shadow} formulation. We prove that the sampled and explicit versions are equivalent under efficient randomized reductions, and that the multiple-shadow variant is computationally equivalent to the single-shadow $\CSV$.

\paragraph{Randomized definition.}

\begin{definition}[Sampled classical shadow]\label{def:sampled}
    A sampled shadow on $n$ qubits is a $4$-tuple $(S,O,A,\chi)$, such that 
    \begin{itemize}
        \item (Shadow) $S$ is an unknown distribution according to which we can sample $\poly(n)$-bit strings,
        \item (Observables) $O=\set{O_i}_{i=1}^m$ is a set of $n$-qubit observables, where $1\leq m\leq 2^{p(n)}$. Given index $i$, a $\poly(n)$-bit description of $O_i$ can be produced in $\poly(n)$-time. Moreover, there exists a $\poly(n)$-time quantum algorithm which, for any $O_i$ and any $n$-qubit state $\rho$, applies measurement $O_i$ to $\rho$.
        \item (Recovery algorithm) $A$ is a $\poly(n)$-time classical algorithm which, given a list $T=(s_1,\dots,s_N)$ of $N=\poly(n)$ samples drawn independently according to $S$, and given $i\in[m]$, produces real number $A(T,i)\in [-1,1]$ within $\chi$ bits of precision.
    \end{itemize}    
\end{definition}

\begin{definition}[Sampled Classical Shadow Validity (SampleCSV)]
The input is a sampled classical shadow $(S,O,A,\chi)$, an integer $N=\poly(n)$, parameters $\alpha,\beta$ satisfying $\beta-\alpha\geq 1/\poly(n)$ and a confidence parameter $0<\delta<1/2$. Let $T=(s_1,\dots,s_N)$ denote a list of $N$ independent samples drawn according to $S$. Decide between the following two cases:
\begin{itemize}
    \item \textbf{Yes}: $\exists$ an $n$-qubit state $\rho$ s.t.$\Pr_{T\sim S^N}\left[\forall i\in[m],\left|\Tr(O_i\rho)-A(T,i)\right|\le\alpha\right]\ge 1-\delta$.
    \item \textbf{No}: $\Pr_{T\sim S^N}\left[\forall\text{ $n$-qubit states }\rho, \exists i\in[m]\text{ such that }\left|\Tr(O_i\rho)-A(T,i)\right| \ge\beta\right]\geq 1-\delta$.
\end{itemize}
As usual, assume without loss of generality that $0\le\alpha<\beta\le2$.
\end{definition}

\begin{lemma}
    $\CSV \leq \SampleCSV$.
\end{lemma}

\begin{proof}
Let $(S,O,A,\chi,\alpha,\beta)$ be a $\CSV$ instance, where
$S=(s_1,\dots,s_L)$ is the given shadow. Fix any $0<\delta<1/2$ with $\log(1/\delta)=\poly(n)$, and set $N:=\left\lceil L(\ln L+\ln(1/\delta))\right\rceil$ .
We construct a $\SampleCSV$ instance $(S',O',A',\chi',N,\alpha',\beta',\delta)$ as follows.

The distribution $S'$ is uniform over the labelled samples $(1,s_1),\dots,(L,s_L)$.
We keep the observables, precision, and parameters unchanged: $O'=O,\; \chi'=\chi,\;\alpha'=\alpha,\; \beta'=\beta$.
The recovery algorithm $A'$ acts as follows. Given a list $T=((j_1,t_1),\dots,(j_N,t_N))$ of $N$ independent samples from $S'$ and an index $i$, it checks whether every label $j\in[L]$ appears at least once. If yes, it keeps one sample
for each label, orders these samples by their labels, discards the labels,
and obtains the original shadow $(s_1,\ldots,s_L)$. It then outputs $A(S,i)$. If some label is missing, $A'$ outputs an arbitrary default value, say $0$.

Let $\mathcal E$ be the event that all labels $1,\dots,L$ appear among
the $N$ samples. By the coupon collector bound,
\[
    \Pr[\mathcal E]\ge 1-Le^{-N/L}\ge 1-\delta .
\]
Conditioned on $\mathcal E$, the algorithm $A'$ reconstructs $S$ exactly, and therefore $A'(T,i)=A(S,i)$ for every $i\in[m]$.

\textbf{Completeness.}
Suppose the original $\CSV$ instance is a YES instance. Then there exists
an $n$-qubit state $\rho$ such that, for all $i\in[m]$, $\left|\Tr(O_i\rho)-A(S,i)\right|\le \alpha$.
Conditioned on $\mathcal E$, for all $i\in[m]$ we have
\[
    \left|\Tr(O'_i\rho)-A'(T,i)\right|
    =
    \left|\Tr(O_i\rho)-A(S,i)\right|
    \le \alpha'.
\]
Since $\Pr[\mathcal E]\ge 1-\delta$, the constructed $\SampleCSV$
instance satisfies the YES condition.

\textbf{Soundness.}
Suppose the original $\CSV$ instance is a NO instance. Then for every
$n$-qubit state $\rho$, there exists an index $i^*$ such that $\left|\Tr(O_{i^*}\rho)-A(S,i^*)\right|\ge \beta$.
Conditioned on $\mathcal E$, the same index $i^*$ satisfies
\[
    \left|\Tr(O'_{i^*}\rho)-A'(T,i^*)\right|
    =
    \left|\Tr(O_{i^*}\rho)-A(S,i^*)\right|
    \ge \beta' .
\]
Since this holds for every $\rho$ and since $\Pr[\mathcal E]\ge 1-\delta$,
the constructed $\SampleCSV$ instance satisfies the NO condition.
\end{proof}
  
\begin{lemma}
    $\SampleCSV\leq_{r}\CSV$.
\end{lemma}
\begin{proof}[Proof sketch]
    If we have a YES (respectively, NO) $\SampleCSV$ instance, sample $N$ strings from distribution $S$, and keep the set of observables and the recovery algorithm the same; this yields a $\CSV$ instance. With probability $\ge1-\delta$, this randomized reduction succeeds, i.e., maps YES (respectively, NO) instances to YES (NO) instances. 
\end{proof}

\paragraph{Multiple shadow consistency.}
The following problem differs from the previous definitions as the input is multiple shadows, and the question is if these shadows can all stem from the same state $\rho$. This is motivated by considering the case where we have (e.g.) two shadows, one of which captures local observable measurements, and the other which targets non-local 
measurements.
\begin{definition}[Multiple classical shadow validity ($\MCSV$)]\label{def:MCSV}
The input is a set of classical shadows $\set{(S_k,O_k,A_k,\chi_k)}_{k=1}^K$ with $K=\poly(n)$, set of parameters $\left(\alpha_k, \beta_k\right)_{k=1}^{K}$ with $k\in[K]$ satisfying $\beta_k-\alpha_k\geq 1/\poly(n)$ decide between the following two cases:
\begin{itemize}
    \item \textbf{Yes}: $\exists$  $n$-qubit state $\rho$ s.t.\ $\forall$ $k\in [K]$, $\forall$ $i\in [M_k]$, 
    $\left| \Tr\left(O_{k,i}\rho\right)-A_k(S_k,i)\right|\leq \alpha_k$.
    \item \textbf{No}: $\forall$ states $\rho$ $\exists$ some $k\in [K]$ and $i\in [M_k]$ s.t.\ $\left| \Tr \left(O_{k,i}\rho\right)-A_k(S_k,i)\right|\geq \beta_k$.
\end{itemize}
Assume succinct access to the observable set and that $0\le\alpha_k<\beta_k\le2\;\;\forall k\in[K]$.
\end{definition}
Though the above definition makes the importance of the $\MCSV$ problem and its difference from $\CSV$ more apparent, it would again be useful for our analysis to define an abstract syntactic variant problem in the same spirit as $\ObsCon$: 

\begin{definition}[Blockwise observable consistency ($\MTObsCon$)]\label{def:BLOC}
The input is $K=\poly(n)$ sets of observables along with their respective expectation values and their tolerance parameters, $\{(O_{k,i}, y_{k,i})_{i=1}^{M_k}, \alpha_k, \beta_k\}_{k=1}^K$ satisfying $\beta_k-\alpha_k\geq 1/\poly(n)$, decide between the following two cases:
\begin{itemize}
    \item \textbf{Yes}: $\exists$  $n$-qubit state $\rho$ s.t.\ $\forall$ $i\in [M_k],k\in[K]$, 
    $\left| \Tr\left(O_{k,i}\rho\right)-y_{k,i}\right|\leq \alpha_k$.
    \item \textbf{No}: $\forall$ $n$-qubit states $\rho$ $\exists$ $i\in [M_k],k\in[K]$ s.t.\ $\left| \Tr \left(O_{k,i}\rho\right)-y_{k,i}\right|\geq \beta_k$.
\end{itemize}
As usual we assume succinct access to both the observables and the expectation values and that $y_{k,i}\in [-1,1]$ and $0\le\alpha_k<\beta_k\le2\;\;\forall k\in[K]$.
  
\end{definition}
Let us now quickly see why these problems are equivalent.
\begin{lemma}\label{lem:MCSV=BLOC}
    $\MCSV$ and $\MTObsCon$ are equivalent under polynomial-time many-one reductions.
\end{lemma}
\begin{proof}
Both directions are straightforward
\begin{itemize}
\item $\MCSV\le\MTObsCon$. Keep the same observables and define $y_{k,i}:=A_k(S_k,i)$.
\item $\MTObsCon\le\MCSV$. Keep the same observables, use dummy shadows $S_k$ and a recovery algorithm that ignores $S_k$ and outputs $y_{k,i}$. 
\end{itemize}
\end{proof}

\begin{definition}
    [$\MTObsCon_{\poly}$] Same as in \cref{def:BLOC} with $M_k=\poly(n)\;\forall k$.
\end{definition}

\begin{lemma}\label{lem:ObsCon-in-BLOC}
$\ObsConPoly\le\MTObsCon_{\poly}$.
\end{lemma}
\begin{proof}
    This is trivial because $\ObsConPoly$ is just a special case of $\MTObsCon_{\poly}$ with $K=1$.
\end{proof}

\begin{lemma}\label{lem:BLOC-in-ObsCon}
$\MTObsCon_{\poly}\le \ObsConPoly$. 
\end{lemma}

\begin{proof}
A $\MTObsCon$ instance has $K=\poly(n)$ blocks indexed by $k\in[K]$, where each block provides pairs $\{(O_{k,i},y_{k,i})\}_{i=1}^{M_k}$ on $n$ qubits and a tolerance parameter pair $(\alpha_k,\beta_k)$ with $\beta_k-\alpha_k \ge 1/\poly(n)$. \\
\\
\textbf{Parameters:} Let
\[
g:=\min_k(\beta_k-\alpha_k),\;\;\tau=g/4,\;\;\alpha_k':=\max\{\alpha_k, \tau\},\;\;t_k:=\tau/\alpha_k'
\]
\textbf{Mapping:}
The reduction outputs the following $\ObsConPoly$ instance:
\[
O_j':=t_kO_{k,i},\;\;y_j':=t_ky_{k,i},\;\;\alpha:=\tau,\;\;\beta:=\tau+\frac{\tau g}{4}.
\]
Where $j\in[M_{\text{tot}}]$ is defined as $j:=S_{k-1}+i$ for $S_k:=\sum_{t=1}^{k}M_t$ and $M_{\text{tot}}:=\sum_{k=1}^KM_k$. Notice that these parameters ensure $\beta-\alpha= g^2/16\ge1/\poly(n)$.\\
\\
\textbf{Completeness:}
$\exists\;\rho$ such that $\forall\;k,i$
$|\Tr(O_{k,i}\rho)-y_{k,i}|\le \alpha_k \le \alpha'_k$\\
\\
Multiplying by $t_k$ gives us:
\[
|\Tr(O'_j\rho)-y'_j|
= |t_k\,\Tr(O_{k,i}\rho)-t_k\,y_{k,i}|
\le t_k \alpha'_k= \tau= \alpha
\]
so the constructed $\ObsConPoly$ instance is a YES instance.\\
\\
\textbf{Soundness:}
 $\forall\rho$ there exists some pair $(k,i)$ for which
\[
|\Tr(O_{k,i}\rho)-y_{k,i}|\ge \beta_{k}\ge \alpha_k+g\ge \alpha_k'-\tau+g=\alpha_k'+\frac{3g}{4}>\alpha_k'+\frac{g}{2}
\]
where the second inequality holds because $g\le\beta_k-\alpha_k$ by definition and the third inequality because $\alpha_k'\le \alpha_k+\tau$, again by definition.\\
Multiplying by $t_k$ gives:
\[
|\Tr(O'_j\rho)-y'_j|= t_{k}|\Tr(O_{k,i}\rho)-y_{k,i}|\ge \tau+\frac{\tau g}{2\alpha_k'} \ge\tau+\frac{\tau g}{4}=\beta.
\]
where the last inequality follows because $\alpha_k'\le 2$ by definition. The constructed $\ObsConPoly$ instance is a NO instance.

\end{proof}
\begin{corollary}\label{cor:BLOC-QMAcomplete}
    $\MTObsCon_{\poly}$ is $\QMA$-complete.
\end{corollary}
\begin{proof}
    Follows from \cref{cor:obsconpoly,lem:BLOC-in-ObsCon,lem:ObsCon-in-BLOC}.
\end{proof}
\begin{corollary}\label{cor:mcsvpoly_qma}
    $\MCSV_{\poly}$ is $\QMA$-complete.
\end{corollary}
\begin{proof}
    Follows from \cref{lem:MCSV=BLOC,cor:BLOC-QMAcomplete}.
\end{proof}
\begin{definition}
    [$\MTObsCon_{\exp}$] Same as in \cref{def:BLOC} with $M_k=\exp(n)$ for some $k$.
\end{definition}
\begin{lemma}\label{lem:BLOC=ObsCon_exp}
    $\ObsConExp\le \MTObsCon_{\exp}\le\ObsConExp$.
\end{lemma}
\begin{proof}
For the first reduction, the argument is again that $\ObsConExp$ is the special case of $\MTObsCon_{\exp}$ with $K=1$. The second reduction follows exactly as in \cref{lem:BLOC-in-ObsCon}, since there we presented a mapping that runs in polynomial time and can be bootstrapped in the succinct framework of the exponential cases. Given indices $k,i$ we get in polynomial time the $(k,i)$-th instance of $\MTObsCon_{\exp}$ (succinct access assumption), say $(O_{k,i},y_{k,i},\alpha_k,\beta_k)$, apply the poly time map given in \cref{lem:BLOC-in-ObsCon}  and get the $\ObsConExp$ instance $(O_j',y_j',\alpha,\beta)$. This concludes the reduction.
\end{proof}
\begin{corollary}\label{cor:BLOC_qccomplete}
    $\MTObsCon_{\exp}$ is $\qc$-complete.
\end{corollary}
\begin{proof}
    Follows from \cref{cor:ObsConexp-SQMAexp_hard,cor:qc=SQMA_exp,lem:BLOC=ObsCon_exp}.
\end{proof}
\begin{corollary}\label{cor:mcsvexp_qc}
    $\MCSV_{\exp}$ is $\qc$-complete.
\end{corollary}
\begin{proof}
    Follows from \cref{lem:MCSV=BLOC,cor:BLOC_qccomplete}.
\end{proof}
\subsection*{Acknowledgements}
The authors thank Asad Raza for helpful discussions.

SG was supported by the DFG under grant numbers 563388236 (Priority Programme ``Quantum Software, Algorithms and Systems – Concepts, Methods and Tools for the Quantum Software Stack'' (SPP 2514)), and 450041824, the BMFTR within the funding program ``Quantum Technologies - from Basic Research to Market'' via project PhoQuant (grant number 13N16103), and the project ``PhoQC'' from the programme ``Profilbildung 2020'', an initiative of the Ministry of Culture and Science of the State of North Rhine-Westphalia.

JE was supported by the BMFTR (QSolid, Hybrid++, QuSol, 
MUNIQC-Atoms, PasQuops), the Munich Quantum Valley (K-4 and K-8), the Quantum Flagship (PasQuans2, Millenion), QuantERA (HQCC), the  Clusters of Excellence MATH+ and ML4Q, the DFG (CRC 183, SPP 2514), Berlin Quantum, and the ERC (DebuQC).

\printbibliography 
 \appendix
\section{$\CLDM\le \ObsConPoly$ }\label{sec:A}
\begin{definition}[Consistency of local density matrices problem (CLDM)\cite{BG22}]
\label{def:cldm}
Let $n\in\mathbb{N}$. The input consists of $((C_1,\rho_1),\ldots,(C_m,\rho_m))$ where
$C_i\subseteq [n]$ and $|C_i|\le k$, and $\rho_i$ is a density matrix on $|C_i|$ qubits
(whose entries are given to $\poly(n)$ precision). Given two parameters $\alpha'$ and $\beta'$,
decide which of the following holds:

\medskip
\noindent\textbf{Yes.} $\exists$ an $n$-qubit quantum state $\tau$ such that for every $i\in[m]$, $\left\|\, \Tr_{\overline{C_i}}(\tau) - \rho_i \,\right\|_{\Tr} \;\le\; \alpha'$ .

\noindent\textbf{No.} $\forall$ $n$-qubit quantum state $\tau$, there exists some $i\in[m]$ such that $\left\|\, \Tr_{\overline{C_i}}(\tau) - \rho_i \,\right\|_{\Tr} \;\ge\; \beta'$.

\end{definition}
\begin{lemma}[Lemma 3.3 \cite{BG22}]\label{lem:cldm-in-qma} The consistency of local density matrices problem is in $\QMA$ for any $k = O(\log n)$, and
$\alpha', \beta'$ such that
$
\epsilon \;:=\; \frac{\beta'}{4^{k}} - \alpha' \;\ge\; \frac{1}{\poly(n)} \, .
$
\end{lemma}

\begin{lemma}
    $\CLDM\le \ObsConPoly$.
\end{lemma}
\begin{proof}
    The mapping is straightforward. For all $i\in[m]$ define $P_j\in P_{|C_i|}$, where $j\in[4^{|C_i|}]$, meaning all the Pauli matrices acting non trivially on the qubits in the $C_i\subseteq[n]$. Then we have
    \[
    O_{i,j}=P_j\in P_{|C_i|}, \;\; y_{i,j}=\Tr(P_j\rho_i),\;\; \alpha=\alpha', \;\; \beta :=\frac{\beta'}{4^k} 
    \]
    \textbf{Completeness:} In the YES case we have 
    \[
    \exists \;\;\text{$n$-qubit state} \;\; \tau \;\;\text{s.t.}\;\;\forall\;i\in[m]\;\; \left\|\, \Tr_{\overline{C_i}}(\tau) - \rho_i \,\right\|_{\Tr} \;\le\; \alpha'  
    \]
    which implies
    \[
    \exists \;\;\text{$n$-qubit state} \;\; \tau \;\;\text{s.t.}\;\;\forall\;i\in[m]\;\;\left| \Tr((P^{C_i}\otimes I^{\overline{C_i}})\tau)-\Tr(P\rho_i)    \right|\le \alpha' .
    \]
    After the mapping we get
    \[
    \exists \;\;\text{$n$-qubit state} \;\; \tau \;\;\text{s.t.}\;\;\forall i\in[m],j\in[4^{|C_i|}]\;\left| \Tr(O_{i,j}\tau)-y_{i,j}    \right|\le \alpha.   
    \]
\textbf{Soundness:} In the NO case we have
    \[\forall \;\;\text{$n$-qubit state} \;\; \tau \;\;\exists\;i\in[m]\;\;\text{s.t.}\;\; \left\|\, \Tr_{\overline{C_i}}(\tau) - \rho_i \,\right\|_{\Tr} \;\ge\; \beta'  
    \]
    which implies  
 \[
    \forall \;\;\text{$n$-qubit state} \;\; \tau \;\;\exists\;i\in[m]\;\;\text{s.t.}\;\;\left| \Tr((P^{C_i}\otimes I^{\overline{C_i}})\tau)-\Tr(P\rho_i)    \right|\ge \frac{\beta'}{4^{|C_i|}}\ge \frac{\beta'}{4^k}.
    \]
    After the mapping we get
    \[
    \forall \;\;\text{$n$-qubit state} \;\; \tau\;\;\exists \;i\in[m],j\in[4^{|C_i|}]\;\text{s.t.}\;\;\left| \Tr(O_{i,j}\tau)-y_{i,j}    \right|\ge \beta.
    \]   
\end{proof}

\section{$\iDCLDM$ on 8-level qudits is $\QMA$-complete}\label{sec:B}
Here we specialize the \cite{BG22} framework for simulatable history states and $\cfont{CLDM}$ hardness to the \cite{HNN13} 1D $d=8$ nearest-neighbor architecture. Concretely, we instantiate BG's simulatable verifier $V_x^{(s)}$ and their snapshot and interval simulation lemmas inside HNN's marker/work formalism and 2-local rule set, yielding a Karp reduction from any $L\in\QMA$ to an $\iDCLDM$ instance supported on single sites and edges of the HNN chain. Our proof follows BG's Theorem 3.4 and Lemma 3.5 (simulation of history states) at the level of local work states, while swapping Kitaev's unary-clock picture for HNN's timetable on an 8-state line.\\
\\
We begin by introducing some useful notation closely related to that of HNN:
Each site $j\in\{1,\dots,N\}$ from the 8-level qudit  chain in 
Ref.\ \cite{HNN13} has a local Hilbert space:
\[
\mathcal H_j
=\bigoplus_{k=1}^{4}\Big(\,\ket{M_k}\otimes \mathbb C\,\Big)\oplus\Big(\,\ket{Q}\otimes \mathbb C^{2}\,\Big)\oplus\Big(\,\ket{Q'}\otimes \mathbb C^{2}\,\Big).
\]
Here the marker symbols $M_k$ carry \emph{no} work qubit (work space $\mathbb C$),
while $Q$ and $Q'$ each carry \emph{one} work qubit (work space $\mathbb C^{2}$).

\paragraph{Time-$t$ configuration and snapshot.}
At step $t$ the HNN timetable fixes a marker string
\[
m(t)=(m_1(t),\dots,m_N(t))\in\{M_1,M_2,M_3,M_4,Q,Q'\}^{N},
\]
and the computation has a work (data) state $\ket{w(t)}$ on the tensor product of the
single-qubit spaces at sites with $m_j(t)\in\{Q,Q'\}$.
The full \emph{snapshot} (no history superposition) is
\[
\ket{\Psi_t} \;=\; \Big(\bigotimes_{j=1}^{N} \ket{m_j(t)}\Big)\ \otimes\ \ket{w(t)}.
\]
Notice that here the allowed marker strings are those specified in
Ref.\ \cite{HNN13}.
\paragraph{Local notation (sites and edges).}
For a site $i$ and an edge $e=(i,i+1)$ define the marker kets and projectors
\[
\ket{c_t(i)} := \ket{m_i(t)},
\qquad
\Pi^{\mathrm{mk}}_{t,i} := \ket{c_t(i)}\bra{c_t(i)},
\]
\[
\ket{c_t(e)} := \ket{m_i(t)} \otimes \ket{m_{i+1}(t)},
\qquad
\Pi^{\mathrm{mk}}_{t,e} := \ket{c_t(e)}\bra{c_t(e)}.
\]

\paragraph{Local work space at time $t$.}
For $S\subseteq\{1,\dots,N\}$ (we will use $S=\{i\}$ or $S=e=(i,i+1)$), set
\[
\mathsf W_S(t)\;=\;\bigotimes_{j\in S}
\begin{cases}
\mathbb C & \text{if } m_j(t)\in\{M_1,M_2,M_3,M_4\},\\[2pt]
\mathbb C^{2} & \text{if } m_j(t)\in\{Q,Q'\},
\end{cases}
\qquad
d_S(t)\;=\;\dim \mathsf W_S(t)\;=\;2^{\,n_S(t)},
\]
where $n_S(t):=\big|\{\,j\in S:\ m_j(t)\in\{Q,Q'\}\,\}\big|\in\{0,1,2\}$.

\paragraph{Local snapshot reduced state.}
The reduced snapshot on $S$ is
\[
X(x,t,S):=\operatorname{Tr}_{\overline S}\big(\ket{\Psi_t}\bra{\Psi_t}\big)
\;\in\; \big(\text{markers on }S\big)\ \otimes\ \mathsf L\big(\mathsf W_S(t)\big).
\]

\paragraph{Edge cases (dimensions).}
For an edge $e=(i,i+1)$ at time $t$:
\[
\dim \mathsf W_e(t)=
\begin{cases}
1 & \text{if }(m_i(t),m_{i+1}(t))\in\{M_1,\dots,M_4\}\times\{M_1,\dots,M_4\}\quad(MM),\\[2pt]
2 & \text{if exactly one of }m_i(t),m_{i+1}(t)\in\{Q,Q'\}\quad(MQ\ \text{or}\ QM),\\[2pt]
4 & \text{if }m_i(t),m_{i+1}(t)\in\{Q,Q'\}\quad(QQ,QQ',Q'Q,Q'Q').
\end{cases}
\]
\paragraph{Verification circuit.}
Start with a QMA verifier $V_x$, amplified so that completeness error
and soundness error are negligible. Fix a constant BG22 simulator
parameter $s\ge 5$ large enough for all site/edge interval simulations
below, and apply the BG22 simulatable compiler with parameter $s$,
obtaining the encoded verifier $V_x^{(s)}$. Thus Lemma 4.8 of
Ref.~\cite{BG22} gives a deterministic simulator for the reduced work
state on any work-qubit set $Y$ with $|Y|\le 3s+2$.

Next we insert SWAP and identity gates as needed, in order to make the
$V_x^{(s)}$ circuit nearest-neighbor, and apply the HNN construction to
obtain an equivalent computation $\widetilde V_x^{(s)}$ on a line of
$N=\poly(|x|)$ eight-dimensional qudits. The legal HNN configurations
carry the work qubits of the nearest-neighbor circuit $V_x^{(s)}$ in the
$Q,Q'$-type two-dimensional subspaces, and all marker patterns and
work-qubit locations are computable from the HNN timetable. Therefore
every site/edge HNN support, and every bounded HNN interval used below,
refers to only constantly many underlying work qubits. Our choice of $s$
ensures that this number is at most $3s+2$, so the BG22 simulator applies. \\
\\
We now prove that $\widetilde V_x^{(s)}$ is simulatable and that the
simulated marginals have zero combined energy with respect to the local
terms of the HNN circuit-to-Hamiltonian construction.\\
\\
We closely follow the exposition in Ref.\ \cite{BG22}, i.e., we first show simulatability and low energy for every snapshot of the computation on a good witness and for small intervals of the history state.

\begin{lemma}[HNN snapshot simulator]
\label{lem:SimSnap}
Fix a constant BG22 simulator parameter $s\ge 5$ large enough for the
site/edge simulations below. Let $\widetilde V_x^{(s)}$ be the
HNN-embedded verification computation constructed above. For every HNN
time $t$ and every support $S\in\{\{i\},\{i,i+1\}\}$, there is a
deterministic polynomial-time procedure $\SimSnap(x,t,S)$ which outputs
the classical description of an $|S|$-qudit density matrix
$\widehat X(x,t,S)$.

If $x\in A_{\mathrm{yes}}$, then for any good witness
$\widetilde\psi^{(s)}$ accepted by $\widetilde V_x^{(s)}$ with
probability at least $1-\negl(|x|)$, the output satisfies
\[
\left\|
\widehat X(x,t,S)
-
\Tr_{\overline S}\bigl(\ket{\Psi_t}\bra{\Psi_t}\bigr)
\right\|_{\Tr}
\le
\negl(|x|),
\]
where $\ket{\Psi_t}$ is the corresponding legal HNN snapshot at time $t$.

Moreover, the simulated work marginals satisfy the following special
properties needed for the HNN input and output checks:
\begin{enumerate}
\item At $t=0$, whenever $W_0(S)$ contains an initialized ancilla work
qubit $j$, its one-qubit simulated marginal is
\[
\Tr_{W_0(S)\setminus\{j\}}\rho(x,0,W_0(S))
=
\ket 0\bra 0 .
\]

\item At the final output-checking time, on the output-checking marker
sector, the simulated output qubit is the accepting branch
$\ket 1\bra 1$, in the sense of the post-decoding/output clause of the
BG22 simulator.
\end{enumerate}
\end{lemma}

\begin{proof}
The algorithm first computes from the HNN timetable the local marker sector
$\Pi^{\mathrm{mk}}_{t,S}$ and the set $W_t(S)$ of underlying work qubits
carried by $S$ at time $t$. Since $S$ is a site or an edge, $|W_t(S)|\le 2\le 3s+2$, so Lemma 4.8 of
Ref.~\cite{BG22} applies to the work-qubit set $W_t(S)$ and outputs a simulated work marginal
$\rho(x,t,W_t(S))$. Define
\[
\widehat X(x,t,S)
:=
\Pi^{\mathrm{mk}}_{t,S}\otimes \rho(x,t,W_t(S)).
\]

Fix a good witness in the YES case. The true HNN snapshot marginal on
$S$ factors as
\[
\Tr_{\overline S}\bigl(\ket{\Psi_t}\bra{\Psi_t}\bigr)
=
\Pi^{\mathrm{mk}}_{t,S}\otimes
\sigma(x,t,W_t(S)),
\]
where $\sigma(x,t,W_t(S))$ is the true reduced work state on $W_t(S)$,
with the convention that it is the scalar density matrix $1$ if
$W_t(S)=\emptyset$. By Lemma 4.8 of Ref.~\cite{BG22},
\[
\left\|
\rho(x,t,W_t(S))
-
\sigma(x,t,W_t(S))
\right\|_{\Tr}
\le
\negl(|x|).
\]
Tensoring with the fixed marker sector is an isometric embedding, and therefore
\[
\left\|
\widehat X(x,t,S)-\Tr_{\overline S}\bigl(\ket{\Psi_t}\bra{\Psi_t}\bigr)\right\|_{\Tr}\le\negl(|x|).
\]

The initialization and post-decoding output properties are exactly the corresponding properties of the BG22 simulated work marginals; inserting the fixed HNN marker sector does not change those one-qubit work marginals.
\end{proof}

\begin{lemma}[HNN marker intervals]\label{lem:HNN-marker-intervals}
For an HNN time $t$, let $m(t)$ denote the legal marker string, and for
a site or nearest-neighbor edge $S$ write $c_t(S):=m(t)|_S$.
For an outside marker pattern $\mu$ on $\overline S$, define
\[
I_\mu(S):=\{t\in\{0,\ldots,T\}:m(t)|_{\overline S}=\mu\}.
\]
Then the nonempty sets $I_\mu(S)$ are contiguous intervals. Moreover,
there is a constant $c_{\mathrm{HNN}}$ such that $|I_\mu(S)|\le c_{\mathrm{HNN}}$
for every site or nearest-neighbor edge $S$ and every outside pattern
$\mu$. Let $\mathcal I(S)$ denote the collection of all nonempty intervals
$I_\mu(S)$.
\end{lemma}
\begin{proof}[Proof sketch]
The reduced marker factor satisfies
\[
\Tr_{\overline S}\bigl(\ket{m(t)}\bra{m(t')}\bigr)
=
\delta_{m(t)|_{\overline S},\,m(t')|_{\overline S}}\,
\ket{c_t(S)}\bra{c_{t'}(S)} .
\]
Hence a cross term survives after tracing out $\overline S$ exactly when
$t$ and $t'$ have the same outside marker pattern. For the legal HNN timetable, the active update is local and moves through the line according to the fixed legal sequence. Since legal marker configurations are not repeated, a fixed outside pattern for a site or edge $S$ can occur only during the constant-size time window in which the
active update is inside the local neighborhood of $S$. Therefore each nonempty $I_\mu(S)$ is contiguous and has size bounded by a constant $c_{\mathrm{HNN}}$.
\end{proof}

\begin{lemma}[HNN interval simulator]
\label{lem:SimInt}
Fix the BG22 simulator parameter $s$ as in \Cref{lem:SimSnap}. Let
$\widetilde V_x^{(s)}$ be the HNN-embedded verification computation
constructed above. Let $S$ be either a site or a nearest-neighbor edge of
the HNN chain, and let $I=\{t_1,t_1+1,\ldots,t_2\}\in\mathcal I(S)$ be one of the marker intervals from \Cref{lem:HNN-marker-intervals}.

There is a deterministic polynomial-time procedure $\SimInt(x,I,S)$
which outputs the classical description of an $|S|$-qudit density matrix
\[
\widehat X(x,I,S)
=
\frac1{|I|}
\sum_{r,r'\in I}
\ket{c_r(S)}\bra{c_{r'}(S)}
\otimes
\rho_{r,r'}(x,I,S),
\]
where $c_r(S)=m(r)|_S$, and $\rho_{r,r'}(x,I,S)$ is the simulated
work-space block associated with the local marker sectors
$c_r(S)$ and $c_{r'}(S)$.

Moreover, if $x\in A_{\mathrm{yes}}$, then for any good witness
$\widetilde\psi^{(s)}$ accepted by $\widetilde V_x^{(s)}$ with
probability at least $1-\negl(|x|)$,
\[
\left\|
\widehat X(x,I,S)-\Tr_{\overline S}\bigl(\Phi_I^{\mathrm{HNN}}\bigr)
\right\|_{\Tr}
\le
\negl(|x|),
\]
where
\[
\Phi_I^{\mathrm{HNN}}
:=
\frac1{|I|}
\sum_{r,r'\in I}
\ket{\Psi_r}\bra{\Psi_{r'}}
=
\frac1{|I|}
\sum_{r,r'\in I}
\ket{m(r)}\bra{m(r')}
\otimes
\ket{w(r)}\bra{w(r')}.
\]
\end{lemma}

\begin{proof}
Let $I=\{t_1,t_1+1,\ldots,t_2\}$. For $r\in I$, write the time-$r$ HNN
snapshot as
\[
\ket{\Psi_r}
=
\ket{m(r)}_{\mathrm{mk}}\otimes\ket{w(r)}_{\mathrm{work}} .
\]

Since $I\in\mathcal I(S)$, the outside marker pattern
$m(r)|_{\overline S}$ is constant for all $r\in I$. Hence, within this interval, if a carried work qubit crossed the boundary of $S$, the outside marker pattern would change. We write $W(S)$ for the
set of underlying work qubits carried by $S$ throughout the interval.

Let $G$ be the set of work qubits touched by gates during the interval, and set
\[
Y:=W(S)\cup G .
\]
Since $S$ is a site or edge and $|I|\le c_{\mathrm{HNN}}$, the set $Y$ has constant size. Our choice of $s$ ensures that $|Y|\le 3s+2$, so the BG22 snapshot simulator applies to the work-qubit set $Y$ at time $t_1$. Let $\widetilde\rho(x,t_1,Y)$
be the simulated work marginal output by that simulator.

For $r,r'\in I$, define the simulated work-space block
\[
\rho_{r,r'}(x,I,S)
:=
\Tr_{Y\setminus W(S)}
\left(
\widetilde U_r\cdots \widetilde U_{t_1+1}\,
\widetilde\rho(x,t_1,Y)\,
\widetilde U_{t_1+1}^{\dagger}\cdots
\widetilde U_{r'}^{\dagger}
\right),
\]
with the convention that if $r=t_1$, then
$\widetilde U_r\cdots \widetilde U_{t_1+1}$ is the identity, and if
$r'=t_1$, then
$\widetilde U_{t_1+1}^{\dagger}\cdots \widetilde U_{r'}^{\dagger}$ is
the identity. The procedure $\SimInt(x,I,S)$ outputs
\[
\widehat X(x,I,S)
=
\frac1{|I|}
\sum_{r,r'\in I}
\ket{c_r(S)}\bra{c_{r'}(S)}
\otimes
\rho_{r,r'}(x,I,S).
\]
This is computable in polynomial time, since $|I|$, $|S|$, and $|Y|$ are
bounded by constants.

Now assume $x\in A_{\mathrm{yes}}$ and fix a good witness
$\widetilde\psi^{(s)}$. Let $\Delta^{(Y)}_{t_1,t_1}$ be the true reduced
work state on $Y$ at time $t_1$. By the BG22 snapshot simulation
guarantee,
\[
\left\|
\widetilde\rho(x,t_1,Y)
-
\Delta^{(Y)}_{t_1,t_1}
\right\|_{\Tr}
\le
\negl(|x|).
\]

By the same construction with the true reduced work state
$\Delta^{(Y)}_{t_1,t_1}$ in place of $\widetilde\rho(x,t_1,Y)$, we have
\[
\Tr_{\overline S}\bigl(\Phi_I^{\mathrm{HNN}}\bigr)
=
\frac1{|I|}
\sum_{r,r'\in I}
\ket{c_r(S)}\bra{c_{r'}(S)}
\otimes
\Tr_{Y\setminus W(S)}
\left(
\widetilde U_r\cdots \widetilde U_{t_1+1}\,
\Delta^{(Y)}_{t_1,t_1}\,
\widetilde U_{t_1+1}^{\dagger}\cdots
\widetilde U_{r'}^{\dagger}
\right).
\]
Indeed, since $I\in\mathcal I(S)$, the marker pattern outside $S$ is
constant on $I$, and hence
\[
\Tr_{\overline S}\bigl(\ket{m(r)}\bra{m(r')}\bigr)
=
\ket{c_r(S)}\bra{c_{r'}(S)}
\qquad
\text{for all } r,r'\in I .
\]

Comparing this expression with the definition of $\widehat X(x,I,S)$, the
only difference is the replacement of $\Delta^{(Y)}_{t_1,t_1}$ by
$\widetilde\rho(x,t_1,Y)$. As in the BG22 interval simulation argument,
the intervening operations are coherent evolution through the interval,
embedding into the local marker sectors, and tracing out registers outside $W(S)$; hence they do not increase trace norm. Therefore
\[
\left\|
\widehat X(x,I,S)
-
\Tr_{\overline S}\bigl(\Phi_I^{\mathrm{HNN}}\bigr)
\right\|_{\Tr}
\le
\left\|
\widetilde\rho(x,t_1,Y)
-
\Delta^{(Y)}_{t_1,t_1}
\right\|_{\Tr}
\le
\negl(|x|).
\]
\end{proof}
We next record the local energy contributions of the interval targets.

\paragraph{Initialization terms.}
Write the initialization Hamiltonian as a sum of one-site marker-activated
penalties,
\[
H_{\mathrm{in}}=\sum_j h^{\mathrm{in}}_j,\;
\text{with}\;\;
h^{\mathrm{in}}_1
=
\big(\ket{Q}\bra{Q}\otimes\ket{1}\bra{1}\big)_1,\quad
h^{\mathrm{in}}_j
=
\big(\ket{Q'}\bra{Q'}\otimes\ket{1}\bra{1}\big)_j .
\]
Let $S_j=\{j\}$ be the support of $h^{\mathrm{in}}_j$. We claim that
$\Tr\left(h^{\mathrm{in}}_j\,\widehat X(x,I,S_j)\right)=0$
for every $j$ appearing in the above sum.

Indeed, the marker part of $h^{\mathrm{in}}_j$ projects onto the
corresponding initial marker sector. Hence, by marker orthogonality, all
blocks of $\widehat X(x,I,S_j)$ vanish against $h^{\mathrm{in}}_j$ except
possibly the block with the initial marker sector. On this surviving
block, the initialization property of \Cref{lem:SimSnap} gives that the
simulated ancilla work qubit is $\ket{0}\bra{0}$. Therefore the
$\ket{1}\bra{1}$ penalty has expectation zero. Thus $\Tr\left(h^{\mathrm{in}}_j\,\widehat X(x,I,S_j)\right)=0$.

\paragraph{Output term.}
Let $S_{\mathrm{out}}=\{j_{\mathrm{out}}\}$ be the one-site support of the
HNN output check, i.e. the final legal marker sector carrying the decoded
output qubit. In our notation the output Hamiltonian has the form
\[
H_{\mathrm{out}}:=
\big(\ket{Q}\bra{Q}\otimes\ket{0}\bra{0}\big)_{\,j_{\mathrm{out}}},
\]
where $\ket{Q}\bra{Q}$ projects onto the final output-checking marker sector.
This $1$-local, marker-activated projector penalizes output $0$ only at the final layer on the
designated active carrier; it is orthogonal (hence contributes $0$) on all other marker sectors.
Let $I\in\mathcal I(S_{\mathrm{out}})$. If $I$ does not contain the final
output-checking time, then the term is inactive by marker orthogonality. If
$I$ contains the final output-checking time, then by the output clause of
\Cref{lem:SimSnap} the simulated output block at $j_{\mathrm{out}}$ is the
accepting branch $\ket{1}\bra{1}$. Hence the local penalty
$\ket{0}\bra{0}$ has expectation zero. Therefore $\Tr\big(H_{\mathrm{out}}\widehat X(x,I,S_{\mathrm{out}})\big)=0$.

\paragraph{Penalty terms:} 
Let $\Sigma=\{M_1,M_2,M_3,M_4,Q,Q'\}$.
Partition the chain into blocks $B_k=\{2n(k-1)+1,\dots,\allowbreak2nk\}$ and, for each edge $e=(i,i+1)$, define its
\emph{location type} $\mathrm{L}(e)\in\{A,B,C,D,E\}$ from the block index and the in block position
(interior odd/even and the two block-end / between-block cases (see Table 5 in Ref.\ \cite{HNN13})).
For each location $L\in\{A,B,C,D,E\}$, let $L_L\subseteq \Sigma^2$ be the set of \emph{legal} adjacent marker
pairs that occur at edges of type $L$ in the timetable, and set $F_L:=\Sigma^2\setminus L_L$, the set of illegal pairs. The penalty Hamiltonian is the edge-local projector
\[
H_{\mathrm{pen}}
=\sum_eH_{\mathrm{pen},e}=
\sum_{e=(i,i+1)}\ \ \sum_{(a,b)\in F_{\mathrm{L}(e)}}
\bigl(\ket{a}\bra{a}\bigr)_i\ \otimes\ \bigl(\ket{b}\bra{b}\bigr)_{i+1}.
\]
(It acts only on marker registers; work qubits, if present, are ignored). Since every marker string $m(r)$ in the HNN timetable is legal, every edge marker pair $c_r(e)$ lies in the legal set $L_{\mathrm L(e)}$. Thus, for every forbidden pair $(a,b)\in F_{\mathrm L(e)}$ and all $r,r'\in I$, the projector $(\ket a\bra a)_i\otimes(\ket b\bra b)_{i+1}$
is orthogonal to the marker block
$\ket{c_r(e)}\bra{c_{r'}(e)}$. Hence $\Tr\big(H_{\mathrm{pen},e}\widehat X(x,I,e)\big)=0$.

The HNN penalty Hamiltonian also contains the boundary marker penalties
$H_{\mathrm{left}}$ and $H_{\mathrm{right}}$. These are one-site marker projectors which exclude symbols that cannot appear at the left and right boundaries of a legal HNN configuration. They are handled identically. Thus
\[
\Tr(H_{\mathrm{left}}\widehat X(x,I,\{1\}))=0,
\quad
\Tr(H_{\mathrm{right}}\widehat X(x,I,\{N\}))=0.
\]
\paragraph{Propagation terms.}
Fix a legal HNN transition $t\to t+1$, and let $e=(i,i+1)$ be the
active edge whose local marker/work pattern changes in this transition. The actual HNN propagation Hamiltonian is written as a sum over rules and
locations, rather than as a separate term for every pair $ (e,t)$. The notation below isolates the legal contribution associated with the transition $t\to t+1$.
Any other local propagation-rule component either corresponds to a marker
transition different from the one realized by this legal time step, or
involves a locally illegal marker pattern detected by the marker-penalty terms. These terms have zero expectation against the target intervals by marker orthogonality.

Define the marker partial isometry on $e$ by
\[
A_{e,t}:=\ket{c_{t+1}(e)}\bra{c_t(e)}.
\]
Thus
\[
A_{e,t}^{\dagger}A_{e,t}=\Pi^{\mathrm{mk}}_{t,e},
\quad
A_{e,t}A_{e,t}^{\dagger}=\Pi^{\mathrm{mk}}_{t+1,e}.
\]
Let $V_{e,t}:\mathsf W_e(t)\longrightarrow \mathsf W_e(t+1)$ be the work-space map induced by this legal HNN transition, under the
canonical identification of the legal marker sectors with the underlying work qubits. Depending on the transition, $V_{e,t}$ is the identity, the prescribed nearest-neighbor gate, a SWAP, or the corresponding move
of the carried work registers. We regard $V_{e,t}$ as an isometry
between the relevant legal work spaces, so $V_{e,t}^{\dagger}V_{e,t}=I_{\mathsf W_e(t)}$.
For bookkeeping, write the legal off-diagonal transition piece as
\[
T_{e,t}:=
A_{e,t}\otimes V_{e,t} + A_{e,t}^{\dagger}\otimes V_{e,t}^{\dagger}.
\]

We also write the two legal diagonal marker checks associated with this
transition as
\[
G^{\mathrm{pre}}_{e,t}
:=
\Pi^{\mathrm{mk}}_{t,S^{\mathrm{pre}}_{e,t}},
\quad
G^{\mathrm{post}}_{e,t}
:=
\Pi^{\mathrm{mk}}_{t+1,S^{\mathrm{post}}_{e,t}},
\]
where $S^{\mathrm{pre}}_{e,t},S^{\mathrm{post}}_{e,t}
\in\bigl\{(i-1,i),(i,i+1),(i+1,i+2)\bigr\}$
are the HNN local supports on which the corresponding pre- and
post-transition marker patterns are checked. Thus the legal grouped
contribution associated with the transition $t\to t+1$ has the form
\[
G^{\mathrm{pre}}_{e,t}
+
G^{\mathrm{post}}_{e,t}
-
T_{e,t}.
\]
The cancellation of this grouped contribution is not claimed at the
level of a single interval. It will be obtained in
\Cref{lem:SimGlobal} after averaging the interval targets with weights $|I|/(T+1)$.

Let
\[
I_{\mathrm{pre}}\in\mathcal I(S^{\mathrm{pre}}_{e,t}),
\quad
I_{\mathrm{post}}\in\mathcal I(S^{\mathrm{post}}_{e,t}),
\quad
I_e\in\mathcal I(e).
\]
For the pre-projector, we have
\[
\Tr\left(
G^{\mathrm{pre}}_{e,t}
\widehat X(x,I_{\mathrm{pre}},S^{\mathrm{pre}}_{e,t})
\right)
=
\frac{1}{|I_{\mathrm{pre}}|}
\sum_{r,r'\in I_{\mathrm{pre}}}
\bra{c_{r'}(S^{\mathrm{pre}}_{e,t})}
\Pi^{\mathrm{mk}}_{t,S^{\mathrm{pre}}_{e,t}}
\ket{c_r(S^{\mathrm{pre}}_{e,t})}
\Tr\left(
\rho_{r,r'}(x,I_{\mathrm{pre}},S^{\mathrm{pre}}_{e,t})
\right).
\]
If \(t\in I_{\mathrm{pre}}\), the only nonzero marker contribution is
$r=r'=t$, and the corresponding diagonal work block has trace $1$.
If \(t\notin I_{\mathrm{pre}}\), every marker block is orthogonal to the
projector. Hence
\[
\Tr\left(
G^{\mathrm{pre}}_{e,t}
\widehat X(x,I_{\mathrm{pre}},S^{\mathrm{pre}}_{e,t})
\right)
=
\begin{cases}
\dfrac{1}{|I_{\mathrm{pre}}|}, & t\in I_{\mathrm{pre}},\\[2mm]
0, & t\notin I_{\mathrm{pre}}.
\end{cases}
\]

Similarly for the post-projector,
\[
\Tr\left(
G^{\mathrm{post}}_{e,t}
\widehat X(x,I_{\mathrm{post}},S^{\mathrm{post}}_{e,t})
\right)
=
\frac{1}{|I_{\mathrm{post}}|}
\sum_{r,r'\in I_{\mathrm{post}}}
\bra{c_{r'}(S^{\mathrm{post}}_{e,t})}
\Pi^{\mathrm{mk}}_{t+1,S^{\mathrm{post}}_{e,t}}
\ket{c_r(S^{\mathrm{post}}_{e,t})}
\Tr\left(
\rho_{r,r'}(x,I_{\mathrm{post}},S^{\mathrm{post}}_{e,t})
\right),
\]
and therefore
\[
\Tr\left(
G^{\mathrm{post}}_{e,t}
\widehat X(x,I_{\mathrm{post}},S^{\mathrm{post}}_{e,t})
\right)
=
\begin{cases}
\dfrac{1}{|I_{\mathrm{post}}|}, & t+1\in I_{\mathrm{post}},\\[2mm]
0, & t+1\notin I_{\mathrm{post}}.
\end{cases}
\]

It remains to record the contribution of the legal off-diagonal piece
$T_{e,t}$ on the active edge $e$. If
$\{t,t+1\}\subseteq I_e$, then by marker orthogonality only the
$(t,t+1)$ and $(t+1,t)$ blocks contribute, and
\[
\Tr\left(
T_{e,t}\widehat X(x,I_e,e)
\right)
=
\frac{1}{|I_e|}
\left(
\Tr\left(V_{e,t}\rho_{t,t+1}(x,I_e,e)\right)
+
\Tr\left(V_{e,t}^{\dagger}\rho_{t+1,t}(x,I_e,e)\right)
\right).
\]
By construction of the interval blocks through the known one-step
evolution on $e$,
\[
\rho_{t+1,t}(x,I_e,e)
=
V_{e,t}\rho_{t,t}(x,I_e,e),
\qquad
\rho_{t,t+1}(x,I_e,e)
=
\rho_{t,t}(x,I_e,e)V_{e,t}^{\dagger}.
\]
Using cyclicity of trace and $V_{e,t}^{\dagger}V_{e,t}=I$, we get
\[
\Tr\left(V_{e,t}\rho_{t,t+1}(x,I_e,e)\right)
=
\Tr\left(
V_{e,t}\rho_{t,t}(x,I_e,e)V_{e,t}^{\dagger}
\right)
=
\Tr\left(\rho_{t,t}(x,I_e,e)\right)
=
1,
\]
and
\[
\Tr\left(V_{e,t}^{\dagger}\rho_{t+1,t}(x,I_e,e)\right)
=
\Tr\left(
V_{e,t}^{\dagger}V_{e,t}\rho_{t,t}(x,I_e,e)
\right)
=\Tr\left(\rho_{t,t}(x,I_e,e)\right)=1.
\]
Therefore
\[
\Tr\left(
T_{e,t}\widehat X(x,I_e,e)
\right)=\frac{2}{|I_e|}
\quad
\text{if } \{t,t+1\}\subseteq I_e.
\]
If $\{t,t+1\}\cap I_e=\varnothing$, then the expectation is $0$ by
marker orthogonality. For the active edge $e$, the mixed case
$|I_e\cap\{t,t+1\}|=1$ cannot occur: the legal transition
$t\to t+1$ changes only the markers on $e$, so $t$ and $t+1$
have the same outside-$e$ marker pattern and hence lie in the same
maximal interval of $\mathcal I(e)$.

We are now ready to prove that $\widetilde V_x^{(s)}$ is simulatable and that the simulations have low-energy with respect to the local terms of the circuit-to-Hamiltonian construction:

\begin{lemma}[Analogous to Lemma 3.5 of Ref.\ \cite{BG22}]
\label{lem:SimGlobal}
Fix the BG22 simulator parameter $s\ge 5$ as in \Cref{lem:SimSnap}.
For any promise problem $A=(A_{\mathrm{yes}},A_{\mathrm{no}})\in\QMA$,
there is a uniform family of HNN-embedded verification computations
$\widetilde V_x^{(s)}$ on a line of $N=\poly(|x|)$ eight-dimensional
qudits, obtained from the BG22-compiled verifier, with the following
property.

There is a deterministic polynomial-time algorithm $\SimHNN(x,S)$ which,
on input $x$ and a support $S\in\bigl\{\{i\}:1\le i\le N\bigr\}
\cup
\bigl\{\{i,i+1\}:1\le i<N\bigr\}$,
outputs the classical description of an $|S|$-qudit density matrix
$\widehat X(x,S)$, with the following properties

\begin{enumerate}
\item If $x\in A_{\mathrm{yes}}$, then there exists a good witness
$\widetilde\psi^{(s)}$ accepted by $\widetilde V_x^{(s)}$ with
probability at least $1-\negl(|x|)$ such that, for every site or edge
support $S$, $\left\|
\widehat X(x,S)
-\Tr_{\overline S}\bigl(\Phi^{\mathrm{HNN}}\bigr)
\right\|_{\Tr}
\le
\negl(|x|)$,
where
\[
\Phi^{\mathrm{HNN}}:=\frac{1}{T+1}\sum_{t,t'=0}^{T}
\ket{\Psi_t}\bra{\Psi_{t'}}
\]
is the HNN history-state density operator for that witness, with
$\ket{\Psi_t}=\ket{m(t)}\otimes\ket{w(t)}$.

\item For every $x$, the HNN local energy contributions are satisfied. Namely, for every non-propagation local term $h$ of the HNN
Hamiltonian, supported on $S_h$, $\Tr\bigl(h\,\widehat X(x,S_h)\bigr)=0$.
Moreover, for every legal HNN transition $t\to t+1$ with active edge $e$, the corresponding grouped legal propagation contribution satisfies $\Tr\bigl(G^{\mathrm{pre}}_{e,t}\widehat X(x,S^{\mathrm{pre}}_{e,t})\bigr)
+\Tr\bigl(G^{\mathrm{post}}_{e,t}\widehat X(x,S^{\mathrm{post}}_{e,t})\bigr)
-\Tr\bigl(T_{e,t}\widehat X(x,e)\bigr)=0$.
\end{enumerate}
\end{lemma}

\begin{proof}
\emph{Construction of $\SimHNN(x,S)$.}
Trace out all marker registers on $\overline S$. For any $t,t'$ the marker factor reduces as
\[
\Tr_{\overline S}\big(\ket{m(t)}\bra{m(t')}\big)
=\Big(\bigotimes_{j\in S}\ket{m_j(t)}\bra{m_j(t')}\Big)\,
\prod_{j\in\overline S}\braket{m_j(t')|m_j(t)}
=\delta_{m(t)\mid_{\overline S},m(t')\mid_{\overline S}}\ket{c_t(S)}\bra{c_{t'}(S)}.
\]
Thus cross-terms vanish unless the outside-$S$ marker pattern agrees at $t$ and $t'$. By \Cref{lem:HNN-marker-intervals}, the times with a fixed outside-$S$
marker pattern form the intervals $\mathcal I(S)=\{I_1,\ldots,I_q\}$,
with $|I_a|\le c_{\mathrm{HNN}}$ for every $a$.

For each $I\in\mathcal I(S)$, run $\SimInt(x,I,S)$ to obtain $\widehat X(x,I,S)$ and output
\[
\widehat X(x,S):=\sum_{I\in\mathcal I(S)} \frac{|I|}{T+1}\widehat X(x,I,S).
\]
This takes deterministic polynomial time. By linearity of partial trace and the partition above,
\[
\Tr_{\overline S}\big(\Phi^{\mathrm{HNN}}\big)
=\sum_{I\in\mathcal I(S)} \frac{|I|}{T+1}\ \Tr_{\overline S}\big(\Phi^{\mathrm{HNN}}_I\big),
\quad
\Phi^{\mathrm{HNN}}_I:=\frac{1}{|I|}\sum_{t,t'\in I}\ket{m(t)}\bra{m(t')}\otimes\ket{w(t)}\bra{w(t')}.
\]
By \Cref{lem:SimInt}, for each $I$ (YES case)
$\big\|\widehat X(x,I,S)-\Tr_{\overline S}(\Phi^{\mathrm{HNN}}_I)\,\big\|_{\Tr}\le\negl(|x|)$.
Averaging with weights $|I|/(T+1)$ and applying the triangle inequality gives
$\big\|\widehat X(x,S)-\Tr_{\overline S}(\Phi^{\mathrm{HNN}})\,\big\|_{\Tr}\le\negl(|x|)$.\\
\\
The second property follows by averaging the local energy contributions of the intervals from before. For non-propagation terms these contributions are $0$ termwise. For propagation, fix a legal HNN transition $t\to t+1$ with active edge
$e$. The corresponding legal grouped contribution has the form
\[
G^{\mathrm{pre}}_{e,t}
+
G^{\mathrm{post}}_{e,t}
-
T_{e,t},
\quad
T_{e,t}:=
A_{e,t}\otimes V_{e,t}+A_{e,t}^{\dagger}\otimes V_{e,t}^{\dagger}.
\]
In the global target, the pre-projector contribution is obtained by
summing over \(\mathcal I(S^{\mathrm{pre}}_{e,t})\):
\[
\sum_{I\in\mathcal I(S^{\mathrm{pre}}_{e,t})}
\frac{|I|}{T+1}
\Tr\left(
G^{\mathrm{pre}}_{e,t}\widehat X(x,I,S^{\mathrm{pre}}_{e,t})\right)
=
\frac1{T+1},
\]
because exactly one interval in \(\mathcal I(S^{\mathrm{pre}}_{e,t})\)
contains \(t\). Similarly, the post-projector contribution is
\[
\sum_{I\in\mathcal I(S^{\mathrm{post}}_{e,t})}
\frac{|I|}{T+1}
\Tr\left(
G^{\mathrm{post}}_{e,t}\widehat X(x,I,S^{\mathrm{post}}_{e,t})
\right)=\frac1{T+1},
\]
because exactly one interval in \(\mathcal I(S^{\mathrm{post}}_{e,t})\)
contains \(t+1\). Finally, the transition contribution is
\[
\sum_{I\in\mathcal I(e)}
\frac{|I|}{T+1}
\Tr\left(
T_{e,t}\widehat X(x,I,e)
\right)=\frac2{T+1},
\]
because the unique interval in $\mathcal I(e)$ containing $t$ also
contains $t+1$. Hence the grouped propagation contribution is
\[
\frac1{T+1}+\frac1{T+1}-\frac2{T+1}=0.
\]
\end{proof}

\paragraph{Proof of \cref{Thm:iDCLDM=QMA}}

Since the containment is straightforward we only show hardness explicitly. Let $A=(A_{\cfont{yes}},A_{\cfont{no}})\in\QMA$. Amplify the original
verifier so that the completeness error is negligible, and apply the
BG22/HNN construction from \Cref{lem:SimGlobal}. On input $x$, the reduction computes the local targets $\widehat X(x,S)$
for all HNN supports $S$ which are sites or nearest-neighbor edges. The
reduction outputs the site and edge marginals
\[
\left\{
(\{i\},\widehat X(x,\{i\})) : 1\le i\le N
\right\}
\cup
\left\{
(\{i,i+1\},\widehat X(x,\{i,i+1\})) : 1\le i<N
\right\}.
\]

Let $H_{\mathrm{HNN}}(x)=\sum_{\ell=1}^{L} h_\ell$ be the full weighted HNN Hamiltonian for $\widetilde V_x^{(s)}$,
expanded into its one-site and nearest-neighbor two-site local pieces.
For propagation, we keep the grouping convention from
\Cref{lem:SimInt}: the pre-projector, post-projector, and transition
pieces belonging to the same HNN transition check are evaluated together.
Define $W:=\sum_{\ell=1}^{L}\|h_\ell\|_\infty$. Since the HNN Hamiltonian has polynomially many local terms with
polynomially bounded weights, $W=\poly(|x|)$. Let $\Delta_{\mathrm{HNN}}\ge \frac1{\poly(|x|)}$ denote the HNN NO-instance lower bound on the ground energy.

Let $\eta=\negl(|x|)$ be the local simulation error from
\Cref{lem:SimGlobal}. By amplification and by choosing the simulation precision sufficiently high, assume $\eta\le \frac{\Delta_{\mathrm{HNN}}}{16W}$. Set the $\iDCLDM$ thresholds to be $\alpha:=2\eta,
\; \beta:=\frac{\Delta_{\mathrm{HNN}}}{4W}$. Then $\beta-\alpha\ge 1/\poly(|x|)$.

If $x\in A_{\cfont{yes}}$, then by \Cref{lem:SimGlobal} there is an HNN history state $\Phi^{\mathrm{HNN}}$ such that for every relevant support
$S$, $\left\|\Tr_{\overline S}\left(\Phi^{\mathrm{HNN}}\right)
-\widehat X(x,S)\right\|_{\Tr}\le \eta\le \alpha$.
Thus the produced local-density-matrix instance is a YES instance.

Now suppose $x\in A_{\cfont{no}}$. Assume for contradiction that there
exists a global state $\tau$ whose relevant local marginals are all
$\beta$-close to the targets: $\left\|
\Tr_{\overline S}(\tau)-\widehat X(x,S)
\right\|_{\Tr}
<\beta$.
By \Cref{lem:SimGlobal}, the target marginals have zero combined HNN
energy: the initialization, output, penalty, and boundary terms vanish
termwise, and each legal propagation contribution vanishes in the grouped sense we discussed.
Therefore,
\[
\begin{aligned}
\Tr\left(H_{\mathrm{HNN}}(x)\tau\right)
&=\sum_{\ell=1}^{L}
\Tr\left(
h_\ell\,\Tr_{\overline{S_\ell}}(\tau)
\right)=\sum_{\ell=1}^{L}\Tr\left(
h_\ell\bigl(\Tr_{\overline{S_\ell}}(\tau)-\widehat X(x,S_\ell)\bigr)
\right)\\
&\le\sum_{\ell=1}^{L}\|h_\ell\|_\infty\left\|\Tr_{\overline{S_\ell}}(\tau)-\widehat X(x,S_\ell)
\right\|_{\Tr}\\
&< W\beta =\frac{\Delta_{\mathrm{HNN}}}{4}<
\Delta_{\mathrm{HNN}}.
\end{aligned}
\]
Here $S_\ell$ denotes the site or edge support of $h_\ell$.

This contradicts the HNN soundness lower bound $\lambda_{\min}(H_{\mathrm{HNN}}(x))\ge \Delta_{\mathrm{HNN}}$.
Therefore, for every global state $\tau$, at least one relevant local marginal violates its target by trace distance at least $\beta$. Hence
the produced $\iDCLDM$ instance is a NO instance. This concludes our proof.

\end{document}